\documentclass{vldb}
\usepackage{amsmath,amssymb}
\usepackage{float}
\usepackage{graphicx}
\usepackage{enumerate}
\usepackage{subfigure}
\usepackage[obeyspaces]{url}

\usepackage{multirow}
\usepackage{makecell}
\usepackage{color}
\usepackage{mathtools}
\usepackage{epstopdf}
\usepackage{balance}
\usepackage{times}
\usepackage{comment}

\usepackage[linesnumbered,vlined,ruled]{algorithm2e}

\newcommand{\cut}[1]{}
\newcommand{\shrink}{\vskip -2ex}

\newcommand{\HL}[1]{\color{black}#1}
\newcommand{\blankline}{\vspace*{2ex}}
\newcommand{\remark}[1]{\noindent\textbf{Remark: }#1}

\newtheorem{theorem}{Theorem}
\newtheorem{definition}{Definition}

\newtheorem{example}{Example}
\newtheorem{lemma}{Lemma}

\newtheorem{corollary}{Corollary}
\newtheorem{assumption}{Assumption}

\newtheorem{procedure}{Procedure}

\def\encoding{\,\xRightarrow{\text{encoding}}\,}

\DeclareMathOperator{\code}{code}
\DeclareMathOperator{\cost}{cost}
\DeclareMathOperator{\tree}{tree}
\DeclareMathOperator{\Dom}{Dom}
\DeclareMathOperator{\Null}{null}

\begin{document}

\title{Sampling-Based Query Re-Optimization}

\author{
\alignauthor{
Wentao Wu
\hspace{.8cm}
Jeffrey F. Naughton}
\hspace{.8cm}
Harneet Singh\\
\vspace{.2cm}
       Department of Computer Sciences, University of Wisconsin-Madison\\
       \vspace{.2cm}
       \{wentaowu, naughton, harneet\}@cs.wisc.edu
}

\maketitle



\begin{abstract}

Despite of decades of work, query optimizers still make mistakes on ``difficult'' queries because of bad cardinality estimates, often due to the interaction of multiple predicates and correlations in the data.
In this paper, we propose a low-cost post-processing step that can take a plan produced by the optimizer, detect when it is likely to have made such a mistake, and take steps to fix it.
Specifically, our solution is a sampling-based iterative procedure that requires almost no changes to the original query optimizer or query evaluation mechanism of the system.
We show that this indeed imposes low overhead and catches cases where three widely used optimizers (PostgreSQL and two commercial systems) make large errors.

%

\end{abstract} 

\section{Introduction}

Query optimizers rely on decent cost estimates of query plans. Cardinality/selectivity estimation is crucial for the accuracy of cost estimates. Unfortunately, although decades of research has been devoted to this area and significant progress has been made, cardinality estimation remains challenging. In current database systems, the dominant approach is to keep various statistics, primarily histograms, about the data.
While histogram-based approaches have worked well for estimating selectivities of local predicates (i.e., predicates over a column of a base table), query optimizers still make mistakes on ``difficult'' queries, often due to the interaction of multiple predicates and correlations in the data~\cite{Lohman-critique}.

Indeed, there is a great deal of work in the literature exploring selectivity estimation techniques beyond histogram-based ones (see Section~\ref{sec:relatedwork}).
Nonetheless, histogram-based approaches remain dominant in practice because of its low overhead.
Note that, query optimizers may explore hundreds or even thousands of candidates when searching for an optimal query plan, and selectivity estimation needs to be done for each candidate.
As a result, a feasible solution has to improve cardinality estimation quality without significantly increasing query optimization time.

In this paper, we propose a low-cost post-processing step that can take a plan produced by the optimizer, detect when it is likely to have made such a mistake, and take steps to fix it.
Specifically, our solution is a sampling-based iterative procedure that requires almost no changes to the original query optimizer or query evaluation mechanism of the system.
We show that this indeed imposes low overhead and catches cases where three widely used optimizers (PostgreSQL and two commercial systems) make large errors.

In more detail, sampling-based approaches (e.g.,~\cite{Charikar-sample00,Haas-sample96,Lipton-sample90}) automatically reflect correlation in the data and between multiple predicates over the data, so they can provide better cardinality estimates on correlated data than histogram-based approaches.
However, sampling also incurs higher overhead.
In previous work~\cite{WuCHN13,WuCZTHN13,WuWHN14}, the authors investigated the effectiveness of using sampling-based cardinality estimates to get better query running time predictions.
The key observation is the following: while it is infeasible to use sampling for all plans explored by
the optimizer, it is feasible to use sampling as a ``post-processing'' step after the search is finished to detect potential errors in optimizer's original cardinality estimates for the final chosen plan.

Inspired by this observation, our basic idea is simple: if significant cardinality estimation errors are detected, the optimality of the returned plan is then itself questionable, so we go one step further to let the optimizer re-optimize the query by also feeding it the cardinality
estimates refined via sampling.
This gives the optimizer second chance to generate a different, perhaps better, plan.
Note that we can again apply the sampling-based validation step to this new plan returned by the optimizer.
It therefore leads to an iterative procedure based on feedback from sampling: we can repeat this optimization-then-validation loop until the plan chosen by the optimizer does not change.
The hope is that this re-optimization procedure can catch large optimizer errors \emph{before} the system even begins executing the chosen query plan.

A couple of natural concerns arise regarding this simple query re-optimization approach. First, how efficient is it? As we have just said, sampling should not be abused given its overhead. Since we propose to run plans over samples iteratively, how fast does this procedure converge? To answer this question, we conduct a theoretical analysis as well as an experimental evaluation. Our theoretical study suggests that, the expected number of iterations can be bounded by $O(\sqrt{N})$, where $N$ is the number of plans considered by the optimizer in its search space. In practice, this upper bound can rarely happen. Re-optimization for most queries tested in our experiments converges after only a few rounds of iteration, and the time spent on re-optimization is ignorable compared with the corresponding query running time.

Second, is it useful? Namely, does re-optimization really generate a better query plan?
{\HL
This raises the question of how to evaluate the effectiveness of re-optimization.
Query optimizers appear to do well almost all of the time.
But the experience of optimizer developers we have talked to is that there are a small number of ``difficult'' queries that cause them most of the pain.
That is, most of the time the optimizer is very good, but when it is bad, it is very bad.
Indeed, Lohman~\cite{Lohman-critique} recently gave a number of compelling real-world instances of optimizers
that, while presumably performing well overall, make serious mistakes on specific queries and data sets. 
It is our belief that an important area for optimizer research is to focus precisely on these few ``difficult'' queries.

We therefore choose to evaluate re-optimization over those difficult, corner-case queries.
Now the hard part is characterizing exactly what these ``difficult'' queries look like.
This will inevitably be a moving target.
If benchmarks were to contain ``difficult'' queries, optimizers would be forced to handle them, and they would no longer be ``difficult,'' so we cannot look to the major benchmarks for examples.
In fact, we implemented our approach in PostgreSQL and tested it on the TPC-H benchmark database, and we did observe significant performance improvement for certain TPC-H queries (Section~\ref{sec:experiments:tpch}).
However, for most of the TPC-H queries, the re-optimized plans are exactly the same as the original ones.
We also tried the TPC-DS benchmark and observed similar phenomena (Appendix~\ref{appendix:sec:exp:tpcds}).
Using examples of real-world difficult queries would be ideal, but we have found it impossible to find well-known public examples of these queries and the data sets they run on.

It is, however, well-known that many difficult queries are made difficult by correlations in the data --- for example, correlations
between multiple selections, and more likely correlations between selections and joins~\cite{Lohman-critique}.
This is our target in this paper.
It is in fact very easy to generate examples of these queries and data sets that confuse all the optimizers (PostgreSQL and two commercial RDBMS) that we tested - such examples are the basis for our ``optimizer torture test'' presented in Section~\ref{sec:benchmark}.
We observed that re-optimization becomes superior on these cases (Section~\ref{sec:experiments:ott}).
While original query plans often take hundreds or even thousands of seconds to finish, after re-optimization all queries can finish in less than 1 second.
We therefore hope that our re-optimization technique can help cover some of those corner cases that are challenging to current query optimizers.}

The idea of query re-optimization goes back to two decades ago (e.g.~\cite{KabraD98,MarklRSLP04}). The main difference between this line of work and our approach is that re-optimization was previously done \emph{after} a query begins to execute whereas our re-optimization is done \emph{before} that. While performing re-optimization during query execution has the advantage of being able to observe accurate cardinalities, it suffers from (sometimes significant) runtime overheads such as materializing intermediate results that have been generated. 
Meanwhile, runtime re-optimization frameworks usually require significant changes to query optimizer's architecture.
Our compile-time re-optimization approach is more lightweight.
The only additional cost is due to running tentative query plans over samples.
The modification to the query optimizer and executor is also limited:
our implementation in PostgreSQL needs only several hundred lines of C code.
Furthermore, we should also note that our compile-time re-optimization approach actually does not conflict with these previous runtime re-optimization techniques: the plan returned by our re-optimization procedure could be further refined by using runtime re-optimization. It remains interesting to investigate the effectiveness of this combination framework.


The rest of the paper is organized as follows. We present the details of our iterative sampling-based re-optimization algorithm in Section~\ref{sec:algorithm}. We then present a theoretical analysis of its efficiency in terms of the number of iterations it requires and the quality of the final plan it returns in Section~\ref{sec:analysis}.
To evaluate the effectiveness of this approach, we further design a database (and a set of queries) with highly correlated data in Section~\ref{sec:benchmark}, and we report experimental evaluation results on this database as well as the TPC-H benchmark databases in Section~\ref{sec:experiments}. We discuss related work in Section~\ref{sec:relatedwork} and conclude the paper in Section~\ref{sec:conclusion}.

\section{The Re-Optimization Algorithm}\label{sec:algorithm}

In this section, we first introduce necessary background information and terminology, and then present the details of the re-optimization algorithm. We focus on using sampling to refine selectivity estimates for join predicates, which are the major source of errors in practice~\cite{Lohman-critique}.
The sampling-based selectivity estimator we used is tailored for join queries~\cite{Haas-sample96}, and it is our goal in this paper to study its effectiveness in query optimization when combined with our proposed re-optimization procedure.
Nonetheless, sampling can also be used to estimate selectivities for other types of operators, such as aggregates (i.e., ``Group By'' clauses) that require estimation of the number of distinct values (e.g.~\cite{Charikar-sample00}). We leave the exploration of integrating other sampling-based selectivity estimation techniques into query optimization as interesting future work.

\subsection{Preliminaries}

In previous work~\cite{WuCHN13,WuCZTHN13,WuWHN14}, the authors used a sampling-based selectivity estimator proposed by Haas et al.~\cite{Haas-sample96} for the purpose of predicting query running times. In the following, we provide an informal description of this estimator.

Let $R_1$, ..., $R_K$ be $K$ relations, and let $R_k^s$ be the sample table of $R_k$ for $1\leq k\leq K$. Consider a join query $q=R_1\bowtie\cdots\bowtie R_K$. The selectivity $\rho_q$ of $q$ can be estimated as
$$\hat{\rho}_q = \frac{|R_1^s\bowtie\cdots\bowtie R_K^s|}{|R_1^s|\times\cdots\times|R_K^s|}.$$
It has been shown that this estimator is both unbiased and strongly consistent~\cite{Haas-sample96}: the larger the samples are, the more accurate this estimator is. Note that this estimator can be applied to joins that are sub-queries of $q$ as well.

\subsection{Algorithm Overview}

As mentioned in the introduction, cardinality estimation is challenging and cardinality estimates by optimizers can be erroneous.
This potential error can be noticed once we apply the aforementioned sampling-based estimator to the query plan generated by the optimizer.
However, if there are really significant errors in cardinality estimates, the optimality of the plan returned by the optimizer can be in doubt.

If we replace the optimizer's cardinality estimates with sampling-based estimates and ask it to re-optimize the query, what would happen? Clearly, the optimizer will either return the same query plan, or a different one. In the former case, we can just go ahead to execute the query plan:
the optimizer does not change plans even with the new cardinalities.
In the latter case, 
the new cardinalities cause the optimizer to change plans.
However, this new plan may still not be trustworthy because the optimizer may still decide its optimality based on erroneous cardinality estimates. To see this, let us consider the following example.
\begin{example}\label{ex:joins-not-covered}
Consider the two join trees $T_1$ and $T_2$ in Figure~\ref{fig:transformations}.
Suppose that the optimizer first returns $T_1$ as the optimal plan.
Sampling-based validation can then refine cardinality estimates for the three joins: $A\bowtie B$, $A\bowtie B\bowtie C$, and $A\bowtie B\bowtie C\bowtie D$.
Upon knowing these refined estimates, the optimizer then returns $T_2$ as the optimal plan.
However, the join $C\bowtie D$ in $T_2$ is not observed in $T_1$ and its cardinality has not been validated.
\end{example}
Hence, we can again apply the sampling-based estimator to this new plan and repeat the re-optimization process. This then leads to an iterative procedure.

Algorithm~\ref{alg:reoptimization} outlines the above idea. Here, we use $\Gamma$ to represent the sampling-based cardinality estimates for joins that have been validated by using sampling. Initially, $\Gamma$ is empty. In the round $i$ ($i\geq 1$), the optimizer generates a query plan $P_i$ based on the current information preserved in $\Gamma$ (line 5). If $P_i$ is the same as $P_{i - 1}$, then we can terminate the iteration (lines 6 to 8). Otherwise, $P_i$ is new and we invoke the sampling-based estimator over it (line 9). We use $\Delta_i$ to represent the sampling-based cardinality estimates for $P_i$, and we update $\Gamma$ by merging $\Delta_i$ into it (line 10). We then move to the round $i + 1$ and repeat the above procedure (line 11).


\begin{algorithm}
  \SetAlgoLined
  \KwIn{$q$, a given SQL query}
  \KwOut{$P_q$, query plan of $q$ after re-optimization}
  \SetAlgoLined
  $\Gamma \leftarrow \emptyset$\;
  $P_0\leftarrow \Null$\;
  $i\leftarrow 1$\;
  \While {true} {
    $P_i\leftarrow GetPlanFromOptimizer(\Gamma)$\;
    \If{$P_i$ is the same as $P_{i-1}$} {
        \textbf{break}\;
    }
    $\Delta_i\leftarrow GetCardinalityEstimatesBySampling(P_i)$\;
    $\Gamma\leftarrow\Gamma\cup\Delta_i$\;
    $i\leftarrow i + 1$\;
  }
  Let the final plan be $P_q$\;
  \Return{$P_q$\;}
  \caption{Sampling-based query re-optimization}
\label{alg:reoptimization}
\end{algorithm}

Note that this iterative process has as its goal improving the selected plan, not
finding a new globally optimal plan. It is certainly possible that the
iterative process misses a good plan because the iterative process
does not explore the complete plan space --- it only explores neighboring transformations of the chosen plan.
Nonetheless, as we will see in Section~\ref{sec:experiments}, this local search is sufficient to catch and
repair some very bad plans.



\begin{figure}
\centering
\includegraphics[width=0.9\columnwidth]{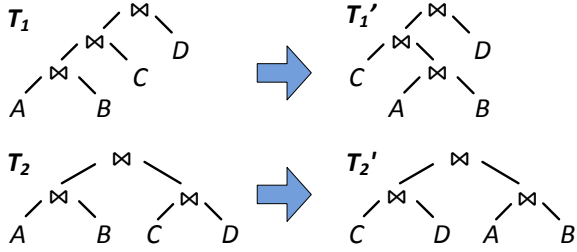}
\caption{Join trees and their local transformations.}
\label{fig:transformations}
\shrink
\end{figure}

\section{Theoretical Analysis} \label{sec:analysis}

In this section, we present an analysis of Algorithm~\ref{alg:reoptimization} from a theoretical point of view. We are interested in two aspects of the re-optimization procedure:
\begin{itemize}
\item \emph{Efficiency}, i.e., how many rounds of iteration does it require before it terminates?
\item \emph{Effectiveness}, i.e., how good is the final plan it returns compared to the original plan, in terms of the cost metric used by the query optimizer?
\end{itemize}
Our following study suggests that (i) the expected number of rounds of iteration in the worst case is upper-bounded by $O(\sqrt{N})$ where $N$ is the number of query plans explored in the optimizer's search space (Section~\ref{sec:analysis:efficiency}); and (ii) the final plan is guaranteed to be no worse than the original plan if
sampling-based cost estimates are consistent with the actual costs (Section~\ref{sec:analysis:optimality}).

\subsection{Local and Global Transformations}

{\HL
We start by introducing the notion of local/global transformations of query plans.
In the following, we use $\tree(P)$ to denote the \emph{join tree} of a query plan $P$.
A join tree is the logical skeleton of a physical plan, which is represented as the set of \emph{ordered} logical joins contained in $P$.
For example, the representation of $T_2$ in Figure~\ref{fig:transformations} is $T_2=\{A\bowtie B, C\bowtie D, A\bowtie B\bowtie C\bowtie D\}$.
}

{\HL
\begin{definition}[Local/Global Transformation]
Two join trees $T$ and $T'$ (of the same query) are \emph{local} transformations of each other if $T$ and $T'$ contain the same set of \emph{unordered} logical joins. Otherwise, they are global transformations.
\end{definition}
}

{\HL
In other words, local transformations are join trees that subject to only exchanges of left/right subtrees.
For example, $A\bowtie B$ and $B\bowtie A$ are different join trees, but they are local transformations.
In Figure~\ref{fig:transformations} we further present two join trees $T'_1$ and $T'_2$ that are local transformations of $T_1$ and $T_2$.
By definition, a join tree is always a local transformation of itself.

Given two plans $P$ and $P'$, we also say that $P'$ is a local/global transformation of $P$ if $\tree(P')$ is a local/global transformation of $\tree(P)$. In addition to potential exchange of left/right subtrees, $P$ and $P'$ may also differ in specific choices of physical join operators (e.g., hash join vs. sort-merge join). Again, by definition, a plan is always a local transformation of itself.
}

\subsection{Convergence Conditions}

At a first glance, even the convergence of Algorithm~\ref{alg:reoptimization} is questionable. Is it possible that Algorithm~\ref{alg:reoptimization} keeps looping without termination? For instance, it seems to be possible that the re-optimization procedure might alternate between two plans $P1$ and $P2$, i.e., the plans generated by the optimizer are $P1$, $P2$, $P1$, $P2$, ...
As we will see, this is impossible and Algorithm~\ref{alg:reoptimization} is guaranteed to terminate. We next present a sufficient condition for the convergence of the re-optimization procedure. We first need one more definition regarding plan coverage.

\begin{definition}[Plan Coverage]\label{definition:plan-coverage}
Let $P$ be a given query plan and $\mathcal{P}$ be a set of query plans. $P$ is \emph{covered} by $\mathcal{P}$ if
$$\tree(P)\subseteq \bigcup\nolimits_{P'\in\mathcal{P}}\tree(P').$$
\end{definition}
That is, all the joins in $\tree(P)$ are included in the join trees of $\mathcal{P}$. As a special case, any plan that belongs to $\mathcal{P}$ is covered by $\mathcal{P}$.

Let $P_i$ ($i\geq 1$) be the plan returned by the optimizer in the $i$-th re-optimization step. We have the following convergence condition for the re-optimization procedure:

\begin{theorem}[Condition of Convergence]\label{theorem:convergence}
Algorithm~\ref{alg:reoptimization} terminates after $n + 1$ ($n\geq 1$) steps if $P_n$ is covered by $\mathcal{P}=\{P_1, ..., P_{n-1}\}$.
\end{theorem}

\begin{proof}
If $P_n$ is covered by $\mathcal{P}$, then using sampling-based validation will not contribute anything new to
the statistics $\Gamma$. That is, $\Delta_n\cup\Gamma = \Gamma$. Therefore, $P_{n+1}$ will be the same as $P_n$, because the optimizer will see the same $\Gamma$ in the round $n+1$ as that in the round $n$. Algorithm~\ref{alg:reoptimization} then terminates accordingly (by lines 6 to 8).
\end{proof}
Note that the convergence condition stated in Theorem~\ref{theorem:convergence} is sufficient by not necessary. It could happen that $P_n$ is not covered by the previous plans $\{P_1,...,P_{n-1}\}$ but $P_{n+1}=P_n$ after using the validated statistics (e.g., if there are no significant errors detected in cardinality estimates of $P_n$).

\begin{corollary}[Convergence Guarantee]\label{corollary:termination}
Algorithm~\ref{alg:reoptimization} is guaranteed to terminate.
\end{corollary}

\begin{proof}
Based on Theorem~\ref{theorem:convergence}, Algorithm~\ref{alg:reoptimization} would only continue if $P_n$ is not covered by $\mathcal{P}=\{P_1, ..., P_{n-1}\}$. In that case, $P_n$ should contain at least one join that has not been included by the plans in $\mathcal{P}$. Since the total number of possible joins considered by the optimizer is finite, $\mathcal{P}$ will eventually reach some fixed point if it keeps growing. The re-optimization procedure is guaranteed to terminate upon that fixed point, again by Theorem~\ref{theorem:convergence}.
\end{proof}

Theorem~\ref{theorem:convergence} also implies the following special case:
\begin{corollary}\label{corollary:local-transformation}
The re-optimization procedure terminates after $n+1$ ($n\geq 1$) steps if $P_n\not\in\mathcal{P}$ but $P_n$ is a local transformation of some $P\in\mathcal{P}=\{P_1, ..., P_{n-1}\}$.
\end{corollary}

{\HL
\begin{proof}
By definition, if $P_n$ is a local transformation of some $P\in\mathcal{P}$, then $\tree(P_n)$ and $\tree(P)$ contain the same set of unordered logical joins. By Definition~\ref{definition:plan-coverage}, $P_n$ is covered by $\mathcal{P}$. Therefore, the re-optimization procedure terminates after $n+1$ steps, by Theorem~\ref{theorem:convergence}.
\end{proof}
}
Also note that Corollary~\ref{corollary:local-transformation} has covered a common case in practice that $P_n$ is a local transformation of $P_{n-1}$.

Based on Corollary~\ref{corollary:local-transformation}, we next present an important property of the re-optimization procedure.
\begin{theorem}\label{theorem:sequence-transformation}
When the re-optimization procedure terminates, exactly one of the following three cases holds:
\begin{enumerate}[(1)]
\item It terminates after 2 steps with $P_2=P_1$.
\item It terminates after $n+1$ steps ($n > 1$). For $1\leq i\leq n$, $P_i$ is a global transformation of $P_j$ with $j < i$.
\item It terminates after $n+1$ steps ($n > 1$). For $1\leq i\leq n - 1$, $P_i$ is a global transformation of $P_j$ with $j < i$, and $P_n$ is a local transformation of some $P\in\mathcal{P}=\{P_1,...,P_{n-1}\}$.
\end{enumerate}
\end{theorem}
That is, when the procedure does not terminate trivially (Case (1)), it can be characterized as a sequence of global transformations with \emph{at most} one more local transformation before its termination (Case (2) or (3)). Figure~\ref{fig:sequence-of-transformations} illustrates the latter two nontrivial cases. A formal proof of Theorem~\ref{theorem:sequence-transformation} is included in Appendix~\ref{appendix:sec:proofs:theorem:sequence-transformation}.

\begin{figure}
\centering
\includegraphics[width=0.9\columnwidth]{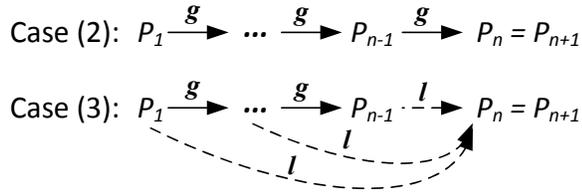}
\caption{Characterization of the re-optimization procedure ($g$ and $l$ stand for global and local transformations, respectively). For ease of illustration, $P_i$ is only noted as a global transformation of $P_{i-1}$, but we should keep in mind that $P_i$ is also a global transformation of all the $P_j$'s with $j < i$.}
\label{fig:sequence-of-transformations}
\shrink
\end{figure}

{\HL
One might raise the question here that why we separate Case (3) from Case (2) in Theorem~\ref{theorem:sequence-transformation}.
Why would the fact that there might be one local transformation be interesting?
The interesting point here is not that we may have one local transformation.
Rather, it is that local transformation can occur at most once, and, if it occurs, it must be the last one in the transformation chain. Actually, we have observed in our experiments that re-optimization of some queries only involves (one) local transformation.
This, however, by no means suggests that local transformations do not lead to any significant optimizations.
For instance, sometimes even just replacing an index-based nested loop join by a hash join (or vice versa) can result in nontrivial performance improvement.
}

\subsection{Efficiency}\label{sec:analysis:efficiency}

We are interested in how fast the re-optimization procedure terminates. As pointed out by Theorem~\ref{theorem:sequence-transformation}, the convergence speed depends on the number of \emph{global} transformations the procedure undergoes. In the following, we first develop a probabilistic model that will be used in our analysis, and then present analytic results for the general case and several important special cases.

\subsubsection{A Probabilistic Model}

Consider a queue of $N$ balls. Originally all balls are not marked. We then conduct the following procedure:

\begin{procedure}\label{procedure:random-ball-model}
In each step, we pick the ball at the head of the queue. If it is marked, then the procedure terminates. Otherwise, we mark it and randomly insert it back into the queue: the probability that the ball will be inserted at the position $i$ ($1\leq i \leq N$) is uniformly $1/N$.
\end{procedure}

We next study the expected number of steps that Procedure~\ref{procedure:random-ball-model} would take before its termination.

\begin{lemma}\label{lemma:expected-steps}
The expected number of steps that Procedure~\ref{procedure:random-ball-model} takes before its termination is:
\begin{equation}\label{eq:SN}
S_N=\sum_{k=1}^N k\cdot(1-\frac{1}{N})\cdots(1-\frac{k-1}{N})\cdot\frac{k}{N}.
\end{equation}
\end{lemma}

We include a proof of Lemma~\ref{lemma:expected-steps} in Appendix~\ref{appendix:sec:proofs:lemma:expected-steps}.
We can further show that $S_N$ is upper-bounded by $O(\sqrt{N})$.

\begin{theorem}\label{theorem:SN-upper-bound}
The expected number of steps $S_N$ as presented in Lemma~\ref{lemma:expected-steps} satisfies $S_N = O(\sqrt{N})$.
\end{theorem}

The proof of Theorem~\ref{theorem:SN-upper-bound} is in Appendix~\ref{appendix:sec:proofs:theorem:SN-upper-bound}.
In Figure~\ref{fig:sN}, we plot $S_N$ by increasing $N$ from $1$ to $1000$.
As we can see, the growth speed of $S_N$ is very close to that of $f(N)=\sqrt{N}$.

\begin{figure}
\centering
\includegraphics[width=\columnwidth]{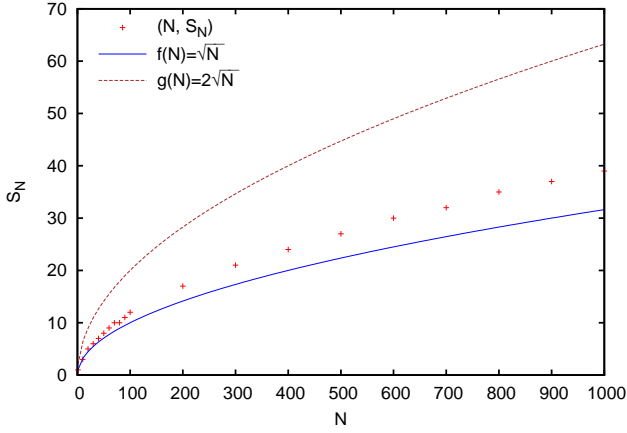}
\caption{$S_N$ with respect to the growth of $N$.}
\label{fig:sN}
\shrink
\end{figure}

%

\subsubsection{The General Case}

In a nutshell, query optimization can be thought of as picking a plan with the lowest estimated cost among a number of candidates. Different query optimizers have different search spaces, so in general we can only assume a search space with $N$ different join trees that will be considered by an optimizer.\footnote{By ``different'' join trees, we mean join trees that are global transformations of each other. We use this convention in the rest of the paper unless specific clarifications are noted.} Let these trees be $T_1$, ..., $T_N$, ordered by their estimated costs. The re-optimization procedure can then be thought of as shuffling these trees based on their refined cost estimates. This procedure terminates whenever (or even before) the tree with lowest estimated cost reoccurs, that is, when some tree appears at the head position for the second time. Therefore, the probabilistic model in the previous section applies here. As a result, by Lemma~\ref{lemma:expected-steps}, the expected number of steps for this procedure to terminate is $S_N$. We formalize this result as the following theorem:

\begin{theorem} \label{theorem:efficiency}
Assume that the position of a plan (after sampling-based validation) in the ordered plans with respect to their costs is uniformly distributed.
Let $N$ be the number of different join trees in the search space.
\footnote{{\HL
Theoretically $N$ could be as large as $O(2^m)$ where $m$ is the number of join operands, assuming that the optimizer uses the bottom-up dynamic programming search strategy.
Nonetheless, this may not always be the case.
For example, some optimizers use the Cascades framework that leverages a top-down search strategy~\cite{Graefe95a}.
An optimizer may even use different search strategies for different types of queries.
For example, PostgreSQL will switch from the dynamic programming search strategy to a randomized search strategy based on genetic algorithm, when the number of joins exceeds a certain threshold (12 by default)~\cite{geqo}.
For this reason, we choose to characterize the complexity of our algorithm in terms of $N$ rather than $m$.
}}
The expected number of steps before the re-optimization procedure terminates is then $S_N$, where $S_N$ is computed by Equation~\ref{eq:SN}.
Moreover, $S_N=O(\sqrt{N})$ by Theorem~\ref{theorem:SN-upper-bound}.
\end{theorem}

We emphasize that the analysis here only targets worst-case performance, which might be too pessimistic. This is because Procedure~\ref{procedure:random-ball-model} only simulates the Case (3) stated in Theorem~\ref{theorem:sequence-transformation}, which is the worst one among the three possible cases. In our experiments, we have found that all queries we tested require less than 10 rounds of iteration, most of which require only 1 or 2 rounds.

\blankline

\remark{The uniformity assumption in Theorem~\ref{theorem:efficiency} may not be valid in practice. It is possible that a plan after sampling-based validation (or, a marked ball in terms of Procedure~\ref{procedure:random-ball-model}) is more likely to be inserted into the front/back half of the queue. Such cases imply that, rather than with an equal chance of overestimation or underestimation, the optimizer tends to overestimate/underestimate the costs of all query plans (for a particular query). This is, however, not impossible.
In practice, significant cardinality estimation errors usually appear locally and propagate upwards. Once the error at some join is corrected, the errors in all plans that contain that join will also be corrected. In other words, the correction of the error at a single join can lead to the correction of errors in many candidate plans. In Appendix~\ref{appendix:sec:analysis}, we further present analysis for two extreme cases: all local errors are overestimates/underestimates.
To summarize, for left-deep join trees, we have the following two results:
\begin{itemize}
\item If all local errors are overestimates, then in the worst case the re-optimization procedure will terminate in at most $m+1$ steps, where $m$ is the number of joins contained in the query.
\item If all local errors are underestimates, then in the worst case re-optimization is expected to terminate in $S_{N/M}$ steps. Here $N$ is the number of different join trees in the optimizer's search space and $M$ is the number of edges in the join graph.
\end{itemize}
Note that both results are better than the bound stated in Theorem~\ref{theorem:efficiency}.
For instance, in the underestimation-only case, if $N=1000$ and $M=10$, we have $S_N=39$ but $S_{N/M}=12$.

{\HL
However, in reality, overestimates and underestimates may coexist. For left-deep join trees, by following the analysis in Appendix~\ref{appendix:sec:analysis}, we can see that such cases sit in between the two extreme cases.
Nonetheless, an analysis for plans beyond left-deep trees (e.g., bushy trees) seems to be challenging.
We leave this as one possible direction for future work.}
}


%

\subsection{Optimality of the Final Plan} \label{sec:analysis:optimality}

We can think of the re-optimization procedure as progressive adjustments of the optimizer's direction when it explores its search space.
The search space depends on the algorithm or strategy used by the optimizer. So does the impact of re-optimization.
But we can still have some general conclusions about the optimality of the final plan regardless of the search space.

\begin{assumption}\label{assumption:cost-function}
The cost estimates of plans using sampling-based cardinality refinement are consistent. That is, for any two plans $P_1$ and $P_2$, if $\cost^s(P_1) < \cost^s(P_2)$, then $\cost^a(P_1) < \cost^a(P_2)$. Here, $\cost^s(P)$ and $\cost^a(P)$ are the estimated cost based on sampling and the actual cost of plan $P$, respectively.
\end{assumption}

We have the following theorem based on Assumption~\ref{assumption:cost-function}. The proof is included in Appendix~\ref{appendix:sec:proofs:theorem:local-optimality}.

\begin{theorem}\label{theorem:local-optimality}
Let $P_1$, ..., $P_n$ be a series of plans generated during the re-optimization procedure. Then $\cost^s(P_n)\leq \cost^s(P_i)$, and thus, by Assumption~\ref{assumption:cost-function}, it follows that $\cost^a(P_n)\leq\cost^a(P_i)$, for $1\leq i \leq n - 1$.
\end{theorem}

That is, the plan after re-optimization is guaranteed to be better than the original plan. Nonetheless, it is difficult to conclude that the plans are improved monotonically, namely, in general it is not true that $\cost^s(P_{i+1})\leq\cost^s(P_i)$, for $1\leq i\leq n - 1$. However, we can prove that this is true if we only have overestimates (proof in Appendix~\ref{appendix:sec:proofs:corollary:local-optimality:overestimates}).

\begin{corollary} \label{corollary:local-optimality:overestimates}
Let $P_1$, ..., $P_n$ be a series of plans generated during the re-optimization procedure. If in the re-optimization procedure only overestimates occur, then $\cost^s(P_{i+1})\leq\cost^s(P_i)$ for $1\leq i\leq n-1$.
\end{corollary}

Our last result on the optimality of the final plan is in the sense that it is the best among all the plans that are local transformations of the final plan (proof in Appendix~\ref{appendix:sec:proofs:theorem:optimality-local-transformation}).

\begin{theorem}\label{theorem:optimality-local-transformation}
Let $P$ be the final plan the re-optimization procedure returns. For any $P'$ such that $P'$ is a local transformation of $P$, it holds that $\cost^s(P)\leq\cost^s(P')$.
\end{theorem}

\subsection{Discussion}

{\HL
We call the final plan returned by the re-optimization procedure the \emph{fixed point} with respect to the initial plan generated by the optimizer. According to Theorem~\ref{theorem:local-optimality}, this plan is a \emph{local optimum} with respect to the initial plan. Note that, if $\mathcal{P} = \{P_1, ..., P_n\}$ covers the whole search space, that is, any plan $P$ in the search space is covered by $\mathcal{P}$, then the locally optimal plan is also globally optimal.
However, in general, it is difficult to give a definitive answer to the question that how far away the plan after re-optimization is from the \emph{true} optimal plan.
It depends on several factors, including the quality of the initial query plan, the search space covered by re-optimization, and the accuracy of the cost model and sampling-based cardinality estimates.
}

A natural question is then the impact of the initial plan. Intuitively, it seems that the initial plan can affect both the fixed point and the time it takes to converge to the fixed point. (Note that it is straightforward to prove that the fixed point must exist and be unique, with respect to the given initial plan.) There are also other related interesting questions. For example, if we start with two initial plans with similar cost estimates, would they converge to fixed points with similar costs as well? We leave all these problems as interesting directions for further investigation.

Moreover, the convergence of the re-optimization procedure towards a fixed point can also be viewed as a validation procedure of the costs of the plans $\mathcal{V}$ that can be covered by $\mathcal{P}=\{P_1, ..., P_n\}$. Note that $\mathcal{V}$ is a subset of the whole search space explored by the optimizer, and $\mathcal{V}$ is induced by $P_1$ --- the initial plan that is deemed as optimal by the optimizer. It is also interesting future work to study the relationship between $P_1$ and $\mathcal{V}$, especially how much of the complete search space can be covered by $\mathcal{V}$.

%

\section{Optimizer ``Torture Test''}\label{sec:benchmark}

Evaluating the effectiveness of a query optimizer is challenging.
As we mentioned in the introduction, query optimizers have to handle not only common cases but also difficult, corner cases.
{\HL
However, it is impossible to find well-known public examples of these corner-case queries and the data sets they run on.
Regarding this, in this section we create our own data sets and queries based on the well-known fact that many difficult queries are made difficult by correlation in the data.} 
We call it ``optimizer torture test'' (OTT), given that our goal is to sufficiently challenge the cardinality estimation approaches used by current query optimizers.
We next describe the details of OTT.

%

\subsection{Design of the Database and Queries}\label{sec:benchmark:design}


Since we target cardinality/selectivity estimation, we can focus on queries that only contain selections and joins.
In general, a selection-join query $q$ over $K$ relations $R_1$, ..., $R_K$ can be represented as
$$q=\sigma_{F}(R_1\bowtie\cdots\bowtie R_K),$$
where $F$ is a selection predicate as in relational algebra (i.e., a boolean formula).
Moreover, we can just focus on equality predicates, i.e., predicates of the form $A=c$ where $A$ is an attribute and $c$ is a constant. Any other predicate can be represented by unions of equality predicates. As a result, we can focus on $F$ of the form
$$F=(A_1=c_1)\land \cdots\land (A_K=c_K),$$
where $A_k$ is an attribute of $R_k$, and $c_k\in\Dom(A_k)$ ($1\leq k\leq K$). Here, $\Dom(A_k)$ is the domain of the attribute $A_k$.

Based on the above observations, our design of the database and queries is as follows:
\begin{enumerate} [(1)]
\item We have $K$ relations $R_1(A_1, B_1)$, ..., $R_K(A_K, B_K)$.
\item We use $A_k$'s for selections and $B_k$'s for joins.
\item Let $R'_k=\sigma_{A_k=c_k}(R_k)$ for $1\leq k\leq K$.
The queries of our benchmark are then of the form:
\begin{equation}\label{eq:query}
R'_1\bowtie_{B_1=B_2}R'_2\bowtie_{B_2=B_3}\cdots\bowtie_{B_{K-1}=B_K}R'_K.
\end{equation}
\end{enumerate}

The remaining question is how to generate data for $R_1$, ..., $R_K$ so that we can easily control the selectivities for the selection and join predicates. This requires us to consider the joint data distribution for $(A_1, ..., A_K, B_1, ..., B_K)$. A straightforward way could be to specify the contingency table of the distribution. However, there is a subtle issue of this approach: we cannot just generate a large table with attributes $A_1$, ..., $A_K$, $B_1$, ..., $B_K$ and then split it into different relations $R_1(A_1, B_1)$, ..., $R_K(A_K, B_K)$.
{\HL
The reason is that we cannot infer the joint distribution $(A_1, ..., A_K, B_1, ..., B_K)$ based on the (marginal) distributions we \emph{observed} on $(A_1, B_1)$, ..., $(A_K, B_K)$.
In Appendix~\ref{appendix:sec:distributions} we further provide a concrete example to illustrate this.
This discrepancy between the observed and true distributions calls for a new approach.
}

\subsection{The Data Generation Algorithm} \label{sec:benchmark:data-gen}

The previous analysis suggests that we can only generate data for each $R_k(A_k, B_k)$ separately and independently, without resorting to their joint distribution. To generate correlated data, we therefore have to make $A_k$ and $B_k$ correlated, for $1\leq k\leq K$.
Because our goal is to challenge the optimizer's cardinality estimation algorithm, we choose to go to the extreme of this direction: let $B_k$ be the same as $A_k$. Algorithm~\ref{alg:data-gen} presents the details of this idea.

\begin{algorithm}
  \SetAlgoLined
  \KwIn{$\Pr(A_k)$, the distribution of $A_k$, for $1\leq k\leq K$}
  \KwOut{$R_k(A_k,B_k)$: tables generated, for $1\leq k\leq K$}
  \SetAlgoLined
  \For{$1\leq k\leq K$} {
    Pick a seed independently for the random number generator\;
    Generate $A_k$ with respect to $\Pr(A_k)$\;
    Generate $B_k=A_k$\;
  }
  \Return{$R_k(A_k,B_k)$, $1\leq k\leq K$\;}
  \caption{Data generation for the OTT database}
\label{alg:data-gen}
\end{algorithm}

We are now left with the problem of specifying $\Pr(A_k)$. While $\Pr(A_k)$ could be arbitrary, we should reconsider our goal of sufficiently challenging the optimizer. We therefore need to know some details about how the optimizer estimates selectivities/cardinalities. Of course, different query optimizers have different implementations, but the general principles are similar. 
In the following, we present the specific technique used by PostgreSQL, which is used in our experimental evaluation in Section~\ref{sec:experiments}.

\subsubsection{PostgreSQL's Approaches} \label{sec:benchmark:data-gen:postgresql}

PostgreSQL stores the following three types of statistics for each attribute $A$ in its \verb|pg_stats| view~\cite{pg-stats}, if the \verb|ANALYZE| command is invoked for the database:
\begin{itemize}
\item the number of distinct values $n(A)$ of $A$;
\item most common values (MCV's) of $A$ and their frequency;
\item an equal-depth histogram for the other values of $A$ except for the MCV's.
\end{itemize}

The above statistics can be used to estimate the selectivity of a predicate over a single attribute in a straightforward manner.
For instance, for the predicate $A=c$ in our OTT queries, PostgreSQL first checks if $c$ is in the MCV's.
If $c$ is present, then the optimizer simply uses the (exact) frequency recorded.
Otherwise, the optimizer assumes a uniform distribution over the non-MCV's and estimates the frequency of $c$ based on $n(A)$.

The approach used by PostgreSQL to estimate selectivities for join predicates is more sophisticated. Consider an equal-join predicate $B_1=B_2$. If MCV's for either $B_1$ or $B_2$ are not available, then the optimizer uses an approach first introduced in System R~\cite{Selinger-AccessPath79} by estimating the reduction factor as $1/\max\{n(B_1),n(B_2)\}$. If, on the other hand, MCV's are available for both $B_1$ and $B_2$, then PostgreSQL tries to refine its estimate by first ``joining'' the two lists of MCV's. For skewed data distributions, this can lead to much better estimates because the join size of the MCV's, which is accurate, will be very close to the actual join size. Other database systems, such as Oracle~\cite{oracle-join-est}, use similar approaches.

To combine selectivities from multiple predicates, PostgreSQL relies on the well-known attribute-value-independence (AVI) assumption, which assumes that the distributions of values of different attributes are independent.

\subsubsection{The Distribution $\Pr(A_k)$ And Its Impact}

From the previous discussion we can see that whether $\Pr(A_k)$ is uniform or skewed will have little difference in affecting the optimizer's estimates if MCV's are leveraged, simply because MCV's have recorded the \emph{exact} frequency for those skewed values. We therefore can just let $\Pr(A_k)$ be uniform. We next analyze the impact of this decision by computing the differences between the estimated and actual cardinalities for the OTT queries.

Let us first revisit the OTT queries presented in Equation~\ref{eq:query}. Note that for an OTT query to be non-empty, the following condition must hold: $B_1=B_2=\cdots =B_{K-1}=B_K$. Because we have intentionally set $A_k=B_k$ for $1\leq k\leq K$, this then implies
\begin{equation} \label{eq:non-empty-condition}
A_1=A_2=\cdots =A_{K-1}=A_K.
\end{equation}
The query size can thus be controlled by the values of the $A$'s. The query is simply empty if Equation~\ref{eq:non-empty-condition} does not hold.
{\HL
In Appendix~\ref{appendix:analysis:ott}, we further present a detailed analysis of the query size when Equation~\ref{eq:non-empty-condition} holds.
To summarize, we are able to control the difference between the query sizes when Equation~\ref{eq:non-empty-condition} holds or not.
Therefore, we can make this gap as large as we wish.
However, the optimizer will give the same estimate of the query size regardless of if Equation~\ref{eq:non-empty-condition} holds or not.
In our experiments (Section~\ref{sec:experiments}) we further used this property to generate instance OTT queries.
}


\section{Experimental Evaluation} \label{sec:experiments}

We present experimental evaluation results of our proposed re-optimization procedure in this section.

\subsection{Experimental Settings} \label{sec:experiments:settings}

We implemented the re-optimization framework in PostgreSQL 9.0.4. The modification to the optimizer is small, limited to several hundred lines of C code, which demonstrates the feasibility of including our framework into current query optimizers. We conducted our experiments on a PC with 2.4GHz Intel dual-core CPU and 4GB memory, and we ran PostgreSQL under Linux 2.6.18.

\subsubsection{Databases and Performance Metrics}

We used both the standard version and a skewed version~\cite{skewed-gen} of the TPC-H benchmark database, as well as our own OTT database described in Section~\ref{sec:benchmark}. 

\paragraph*{TPC-H Benchmark Databases}

We used TPC-H databases at the scale of 10GB in our experiments.
The generator for skewed TPC-H database uses a parameter $z$ to control the skewness of each column by generating Zipfian distributions.
The larger $z$ is, the more skewed the generated data are.
$z=0$ corresponds to a uniform distribution.
In our experiments, we generated a skewed database by setting $z=1$.

\paragraph*{OTT Database}

We created an instance of the OTT database in the following manner. We first generated a standard 1GB TPC-H database. We then extended the 6 largest tables (\emph{lineitem}, \emph{orders}, \emph{partsupp}, \emph{part}, \emph{customer}, \emph{supplier}) by adding two additional columns $A$ and $B$ to each of them. 
As discussed in Section~\ref{sec:benchmark:data-gen}, we populated the extra columns with uniform data.
The domain of a column is determined by the number of rows in the corresponding table: if the table contains $r$ rows, then the domain is $\{0, 1, ..., r / 100 - 1\}$. In other words, each distinct value in the domain appears roughly 100 times in the generated column.
We further created an index on each added column.

\paragraph*{Performance Metrics}

In our experiments, we measured the following performance metrics for each query on each database:
\begin{enumerate} [(1)]
\item the original running time of the query;
\item the number of iterations the re-optimization procedure requires before its termination;
\item the time spent on the re-optimization procedure;
\item the total query running time including the re-optimization time.
\end{enumerate}
Based on studies in the previous work~\cite{WuCZTHN13}, in all of our experiments we set the sampling ratio to be 0.05, namely, $5\%$ of the data were taken as samples.

\subsubsection{Calibrating Cost Models}

The previous work~\cite{WuCZTHN13} has also shown that, after proper calibration of the cost models used by the optimizer, we could have better estimates of query running times.
An interesting question is then: would calibration also improve query optimization?

In our experiments, we also investigated this problem.
Specifically, we ran the offline calibration procedure (details in~\cite{WuCZTHN13}) and replaced the default values of the five cost units (\emph{seq\_page\_cost}, \emph{random\_page\_cost}, \emph{cpu\_tuple\_cost}, \emph{cpu\_index\_tuple\_cost}, and\\ \emph{cpu\_operator\_cost}) in  \verb|postgresql.conf| (i.e., the configuration file of PostgreSQL server) with the calibrated values.
In the following, we also report results based on calibrated cost models.

\subsection{Results on the TPC-H Benchmark}\label{sec:experiments:tpch}

We tested 21 TPC-H queries. (We excluded Q15, which is not supported by our current implementation because it requires to create a view first.) For each TPC-H query, we randomly generated 10 instances. We cleared both the database buffer pool and the file system cache between each run of each query.

\subsubsection{Results on Uniform Database}

Figure~\ref{fig:tpch-skew0:time} presents the average running times and their standard deviations (as error bars) of these queries over the uniform database.

\begin{figure}[t]
\centering
\subfigure[Without calibration of the cost units]{ \label{fig:tpch-skew0:time:no-calib}
\includegraphics[width=\columnwidth]{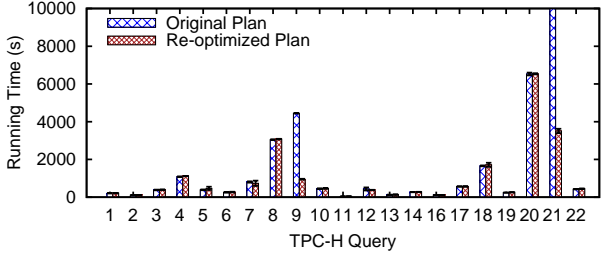}}
\subfigure[With calibration of the cost units]{ \label{fig:tpch-skew0:time:calib}
\includegraphics[width=\columnwidth]{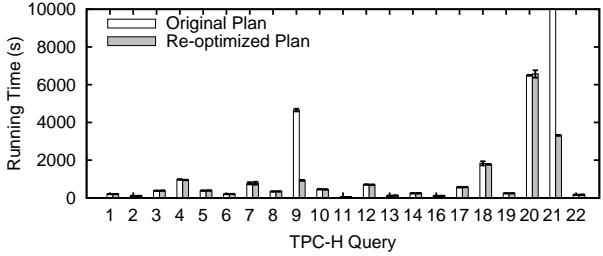}}
\caption{Query running time over uniform 10GB TPC-H database ($z=0$).}
\label{fig:tpch-skew0:time}
\shrink
\end{figure}

We have two observations. First, while the running times for most of the queries almost do not change, we can see significant improvement for some queries. For example, as shown in Figure~\ref{fig:tpch-skew0:time:no-calib}, even without calibration of the cost units, the average running time of Q9 drops from 4,446 seconds to only 932 seconds, a 79\% reduction; more significantly, the average running time of Q21 drops from 20,746 seconds (i.e., almost 6 hours) to 3,508 seconds (i.e., less than 1 hour), a 83\% reduction.

Second, calibration of the cost units can sometimes significantly reduce the running times for some queries. For example, comparing Figure~\ref{fig:tpch-skew0:time:no-calib} with Figure~\ref{fig:tpch-skew0:time:calib} we can observe that the average running time of Q8 drops from 3,048 seconds to only 339 seconds, a 89\% reduction, by just using calibrated cost units without even invoking the re-optimization procedure.

\begin{figure}
\centering
\includegraphics[width=\columnwidth]{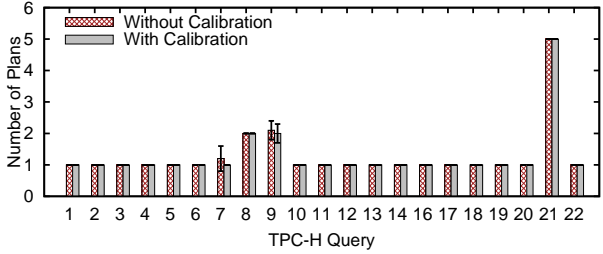}
\caption{The number of plans generated during re-optimization over uniform 10GB TPC-H database.}
\label{fig:tpch-skew0:num}
\shrink
\end{figure}

We further studied the re-optimization procedure itself. Figure~\ref{fig:tpch-skew0:num} presents the number of plans generated during re-optimization. It substantiates our observation in Figure~\ref{fig:tpch-skew0:time}: for the queries whose running times were not improved, the re-optimization procedure indeed picked the same plans as those originally chosen by the optimizer. Figure~\ref{fig:tpch-skew0:time-sampling} further compares the query running time excluding/including the time spent on re-optimization. For all the queries we tested, re-optimization time is ignorable compared to query execution time, which demonstrates the low overhead of our re-optimization procedure.

\begin{figure}[t]
\centering
\subfigure[Without calibration of the cost units]{ \label{fig:tpch-skew0:time-sampling:no-calib}
\includegraphics[width=\columnwidth]{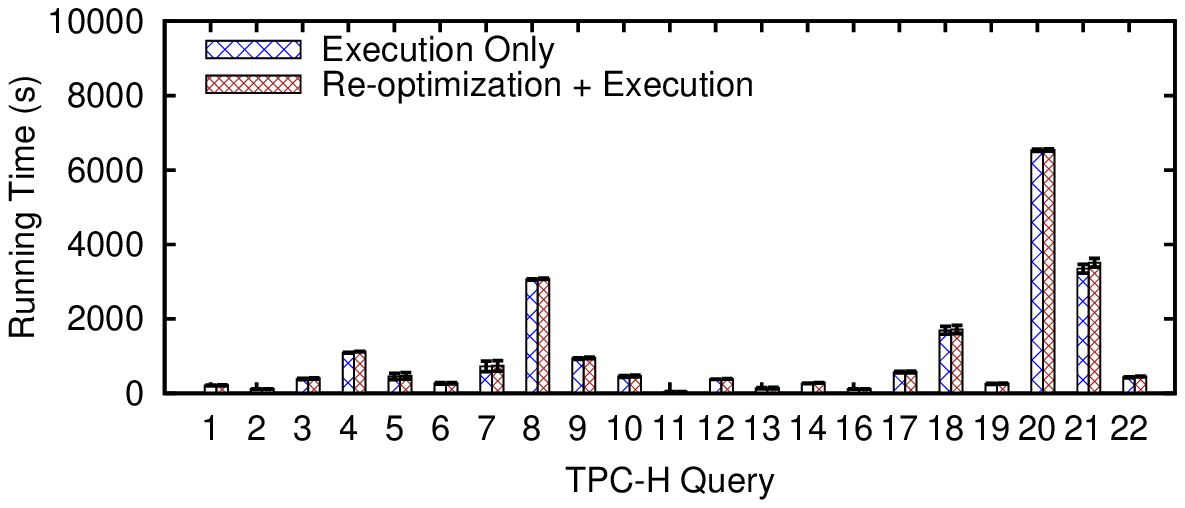}}
\subfigure[With calibration of the cost units]{ \label{fig:tpch-skew0:time-sampling:calib}
\includegraphics[width=\columnwidth]{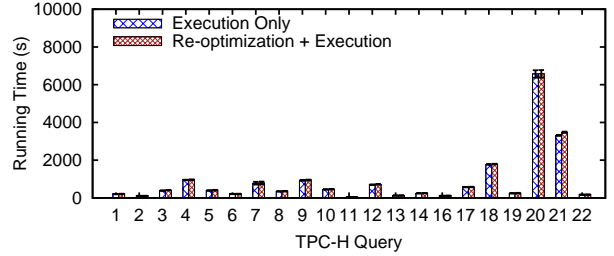}}
\caption{Query running time excluding/including re-optimization time over uniform 10GB TPC-H database ($z=0$).}
\label{fig:tpch-skew0:time-sampling}
\shrink
\end{figure}

\subsubsection{Results on Skewed Database}

On the skewed database, we have observed results similar to that on the uniform database.
Figure~\ref{fig:tpch-skew1:time} presents the running times of the queries,
with or without calibration of the cost units.\footnote{
\HL{We notice that Q17 in Figure~\ref{fig:tpch-skew1:time} has a large error bar.
The error bars represent variance due to different instances of the query template (different
constants in the query). Q17 therefore has a large variance, because we used a skewed TPC-H database.
The variance is much smaller when a uniform database is used (see Figure~\ref{fig:tpch-skew0:time}).}
}
While it looks quite similar to Figure~\ref{fig:tpch-skew0:time}, there is one interesting phenomenon not shown before.
In Figure~\ref{fig:tpch-skew1:time:no-calib} we see that, without using calibrated cost units, the average running times for Q8 and Q9 actually increase after re-optimization.
Recall that in Section~\ref{sec:analysis:optimality} we have shown the local optimality of the plan returned by the re-optimization procedure (Theorem~\ref{theorem:local-optimality}).
However, that result is based on the assumption that sampling-based cost estimates are consistent with actual costs (Assumption~\ref{assumption:cost-function}).
Here this seems not the case.
Nonetheless, after using calibrated cost units, both the running times of Q8 and Q9 were significantly improved (Figure~\ref{fig:tpch-skew1:time:calib}).

\begin{figure}[t]
\centering
\subfigure[Without calibration of the cost units]{ \label{fig:tpch-skew1:time:no-calib}
\includegraphics[width=\columnwidth]{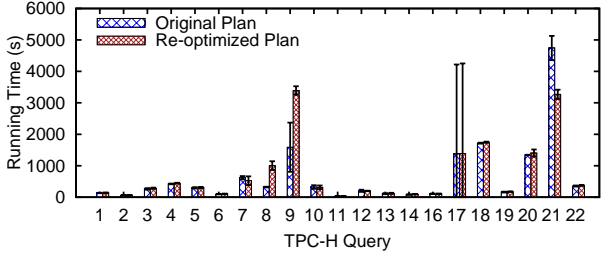}}
\subfigure[With calibration of the cost units]{ \label{fig:tpch-skew1:time:calib}
\includegraphics[width=\columnwidth]{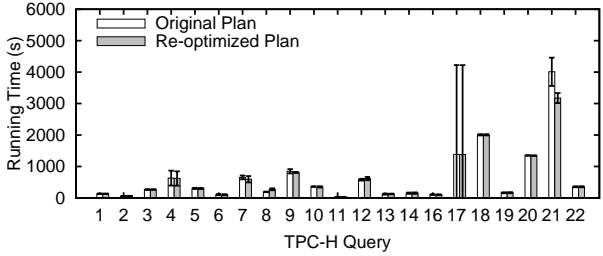}}
\caption{Query running time over skewed 10GB TPC-H database ($z=1$).}
\label{fig:tpch-skew1:time}
\shrink
\end{figure}


\begin{figure}
\centering
\includegraphics[width=\columnwidth]{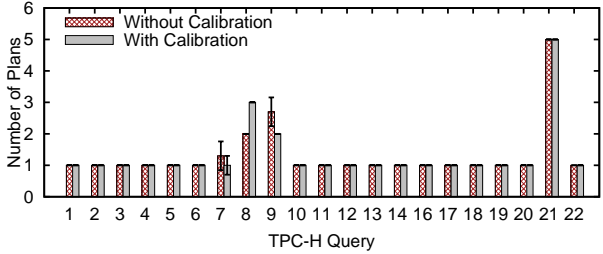}
\caption{The number of plans generated during re-optimization over skewed 10GB TPC-H database.}
\label{fig:tpch-skew1:num}
\shrink
\end{figure}


We further present the number of plans considered during re-optimization in Figure~\ref{fig:tpch-skew1:num}.
Note that re-optimization seems to be more active on skewed data.
Figure~\ref{fig:tpch-skew1:time-sampling} shows the running times excluding/including the re-optimization times of the queries.
Again, the additional overhead of re-optimization is trivial.

\begin{figure}[t]
\centering
\subfigure[Without calibration of the cost units]{ \label{fig:tpch-skew1:time-sampling:no-calib}
\includegraphics[width=\columnwidth]{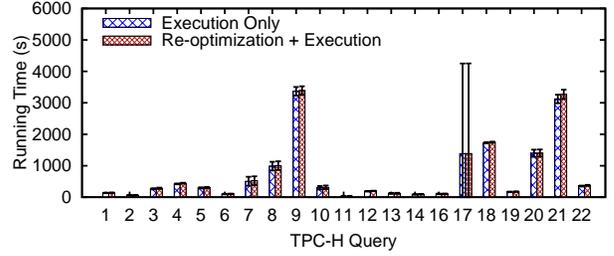}}
\subfigure[With calibration of the cost units]{ \label{fig:tpch-skew1:time-sampling:calib}
\includegraphics[width=\columnwidth]{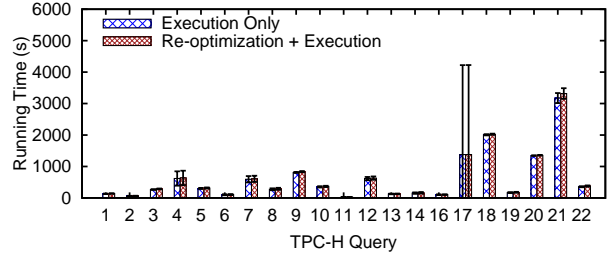}}
\caption{Query running time excluding/including re-optimization time over skewed 10GB TPC-H database ($z=1$).}
\label{fig:tpch-skew1:time-sampling}
\shrink
\end{figure}

\subsubsection{Discussion}

While one might expect the chance for re-optimization to generate a better plan is higher on skewed databases, our experiments show that this may not be the case, at least for TPC-H queries.
There are several different situations, though.
First, if a query is too simple, then there is almost no chance for re-optimization.
For example, Q1 contains no join, whereas Q16 and Q19 involve only one join so only local transformations are possible.
Second, the final plan returned by the re-optimization procedure heavily relies on the initial plan picked by the optimizer, which is the seed or starting point where re-optimization originates.
Note that, even if the optimizer has picked an inefficient plan, re-optimization cannot help if the estimated cost of that plan is not significantly erroneous.
One question is if this is possible: the optimizer picks an inferior plan whose cost estimate is correct?
This actually could happen because the optimizer may (incorrectly) overestimate the costs of the other plans in its search space.
Another subtle point is that the inferior plan might be robust to certain degree of errors in cardinality estimates.
Previous work has reported this phenomenon by noticing that the plan diagram (i.e., all possible plans and their governed optimality areas in the selectivity space) is dominated by just a couple of query plans~\cite{Picasso05}.

In summary, the effectiveness of re-optimization depends on factors that are out of the control of the re-optimization procedure itself.
Nevertheless, we have observed intriguing improvement for some long-running queries by applying re-optimization, especially after calibration of the cost units.

\subsection{Results of the Optimizer Torture Test}\label{sec:experiments:ott}

We created queries following our design of the OTT in Section~\ref{sec:benchmark:design}.
Specifically, if a query contains $n$ tables (i.e., $n-1$ joins), we let $m$ of the selections be $A=0$ ($A=1$), and let the remaining $n-m$ selections be $A=1$ ($A=0$).
We generated two sets of queries: (1) $n=5$ (4 joins), $m=4$; and (2) $n=6$ (5 joins), $m=4$.
Note that the maximal non-empty sub-queries then contain 3 joins over 4 tables with result size of roughly $100^4=10^8$ rows.\footnote{Recall that a non-empty query must have equal $A$'s (Equation~\ref{eq:non-empty-condition}) and we generated data with roughly 100 rows per distinct value (Section~\ref{sec:experiments:settings}).}
However, the size of each (whole) query is 0.
So we would like to see the ability of the optimizer as well as the re-optimization procedure to identify the empty/non-empty sub-queries.

\begin{figure}[t]
\centering
\subfigure[Without calibration of the cost units]{ \label{fig:hdc-4:time:no-calib}
\includegraphics[width=\columnwidth]{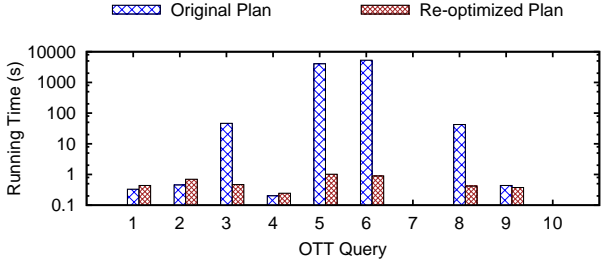}}
\subfigure[With calibration of the cost units]{ \label{fig:hdc-4:time:calib}
\includegraphics[width=\columnwidth]{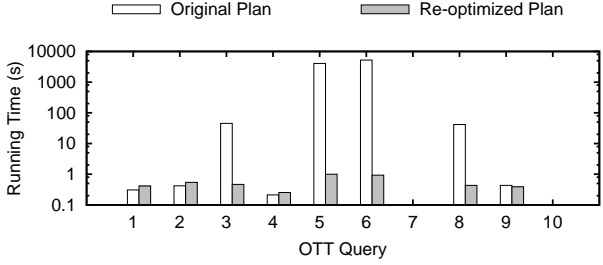}}
\caption{Query running times of 4-join queries.}
\label{fig:hdc-4:time}
\shrink
\end{figure}

Figure~\ref{fig:hdc-4:time} and~\ref{fig:hdc-5:time} present the running times of the 4-join and 5-join queries, respectively.
We generated in total 10 4-join queries and 30 5-join queries.
Note that the $y$-axes are in log scale and we do not show queries that finish in less than 0.1 second.
As we can see, sometimes the optimizer failed to detect the existence of empty sub-queries: it generated plans where empty join predicates were evaluated after the non-empty ones.
The running times of these queries were then hundreds or even thousands of seconds.
On the other hand, the re-optimization procedure did an almost perfect job in detecting empty joins, which led to very efficient query plans where the empty joins were evaluated first: all the queries after re-optimization finished in less than 1 second.

\begin{figure}[!h]
\centering
\subfigure[Without calibration of the cost units]{ \label{fig:hdc-5:time:no-calib}
\includegraphics[width=\columnwidth]{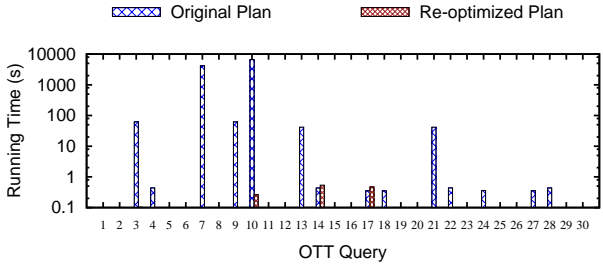}}
\subfigure[With calibration of the cost units]{ \label{fig:hdc-5:time:calib}
\includegraphics[width=\columnwidth]{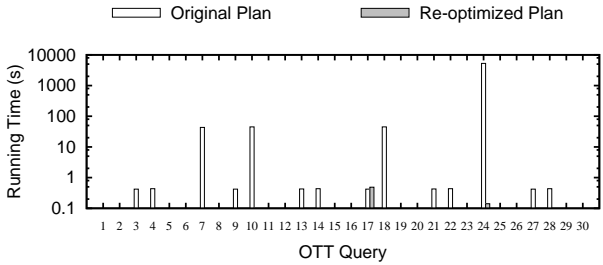}}
\caption{Query running times of 5-join queries.}
\label{fig:hdc-5:time}
\shrink
\end{figure}

We also did similar studies regarding the number of plans generated during re-optimization and the time it consumed.
Due to space constraints, we refer the readers to Appendix~\ref{appendix:sec:exp:ott} for the details.



One might argue that the OTT queries are really contrived: these queries are hardly to see in real-world workloads.
While this might be true, we think these queries serve our purpose as exemplifying extremely hard cases for query optimization.
Note that hard cases are not merely long-running queries: queries as simple as sequentially scanning huge tables are long-running too, but there is nothing query optimization can help with.
Hard cases are queries where efficient execution plans do exist but it might be difficult for the optimizer to find them.
The OTT queries are just these instances.
Based on the experimental results of the OTT queries, re-optimization is helpful to give the optimizer second chances if it initially made a bad decision.

Another concern is if commercial database systems could do a better job on the OTT queries. In regard of this, we also ran the OTT over two major commercial database systems. The performance is very similar to that of PostgreSQL (Figure~\ref{fig:db2:time} and~\ref{fig:sqlserver:time}).
We therefore speculate that commercial systems could also benefit from our re-optimization technique proposed in this paper.

\begin{figure}[t]
\centering
\subfigure[4-join OTT queries]{ \label{fig:db2:time:4-join}
\includegraphics[width=\columnwidth]{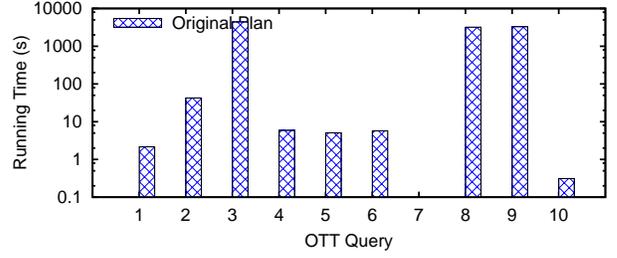}}
\subfigure[5-join OTT queries]{ \label{fig:db2:time:5-join}
\includegraphics[width=\columnwidth]{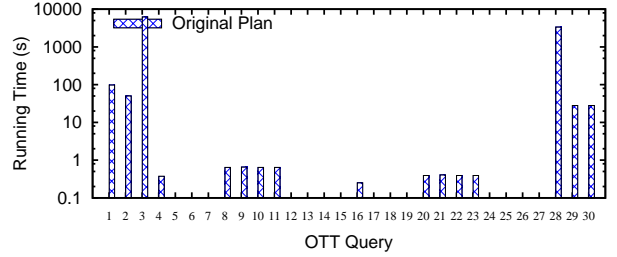}}
\caption{Query running times of the OTT queries on the commercial database system A.}
\label{fig:db2:time}
\shrink
\end{figure}

\begin{figure}[!h]
\centering
\subfigure[4-join OTT queries]{ \label{fig:sqlserver:time:4-join}
\includegraphics[width=\columnwidth]{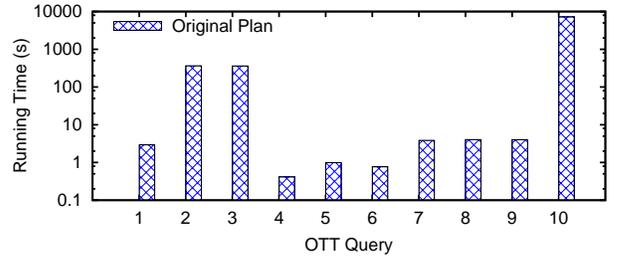}}
\subfigure[5-join OTT queries]{ \label{fig:sqlserver:time:5-join}
\includegraphics[width=\columnwidth]{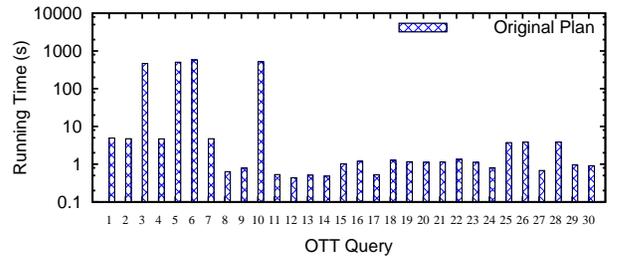}}
\caption{Query running times of the OTT queries on the commercial database system B.}
\label{fig:sqlserver:time}
\shrink
\end{figure}

\subsubsection{A Note on Multidimensional Histograms}\label{sec:experiments:multi-histograms}

Note that even using multidimensional histograms (e.g.,~\cite{BrunoCG01,MuralikrishnaD88,Poosala-hist97}) may not be able to detect the data correlation presented in the OTT queries, unless the buckets are so fine-grained that the exact joint distributions are retained.
To understand this, let us consider the following example.
\begin{example}
Following our design of OTT, suppose that now we only have two tables $R_1(A_1, B_1)$ and $R_2(A_2, B_2)$.
Moreover, suppose that each $A_k$ (and thus $B_k$) contains $m=2l$ distinct values, and we construct (perfect) 2-dimensional histograms on $(A_k, B_k)$ ($k=1,2$).
Each dimension is evenly divided into $\frac{m}{2}=l$ intervals, so each histogram contains $l^2$ buckets.
The joint distribution over $(A_k, B_k)$ estimated by using the histogram is then:
\begin{displaymath}
\left\{
\begin{array}{ll}
\Pr(2r-2\leq A_k < 2r, 2r-2\leq B_k < 2r) = \frac{1}{l}, & \textrm{$1\leq r\leq l$;}\\
\Pr(a_l\leq A_k < a_2, b_1\leq B_k < b_2) = 0, & \textrm{otherwise.}
\end{array}
\right.
\end{displaymath}
For instance, if $m=100$, then $l=50$. So we have
$\Pr(0\leq A_k < 2, 0\leq B_k < 2)=\cdots=\Pr(98\leq A_k < 100, 98\leq B_k < 100)=\frac{1}{50},$
while all the other buckets are empty.
On the other hand, the actual joint distribution is
\begin{displaymath}
\Pr(A_k=a, B_k=b)=\left\{
\begin{array}{ll}
\frac{1}{m}, & \textrm{if $a=b$;}\\
0, & \textrm{otherwise.}
\end{array}
\right.
\end{displaymath}
Now, let us consider the selectivities for two OTT queries:
\begin{enumerate}[$(q_1)$]
\item $\sigma_{A_1=0\land A_2=1\land B_1=B_2}(R_1\times R_2)$;
\item $\sigma_{A_1=0\land A_2=0\land B_1=B_2}(R_2\times R_2)$.
\end{enumerate}
We know that $q_1$ is empty but $q_2$ is not.
However, the estimated selectivity (and thus cardinality) of $q_1$ and $q_2$ is the same by using the 2-dimensional histogram, because of the assumption used by histograms that data inside each bucket is uniformly distributed.\footnote{It is easy to verify that the selectivity estimates are $\hat{s}_1=\hat{s}_2=\frac{1}{8l^2}$.}
With the setting used in our experiments, $m=100$ and thus $l=50$. So each 2-dimensional histogram contains $l^2=2,500$ buckets. However, even such detailed histograms cannot help the optimizer distinguish empty joins from nonempty ones. Furthermore, note that our conclusion is independent of $m$, while the number of buckets increases quadratically in $m$. For instance, when $m=10^4$ which means we have histograms containing $2.5\times 10^7$ buckets, the optimizer still cannot rely on the histograms to generate efficient query plans for OTT queries.
\end{example}

{\HL
\subsection{Effectiveness of Iteration}

One interesting further question is how much benefit iteration brings in.
Since running plans over samples incurs additional cost, we may wish to stop the iteration as early as possible rather than wait until its convergence.
To investigate this, we also tested execution times for plans generated during re-optimization on the original databases.
We focus on queries for which at least two plans were generated.
Figure~\ref{fig:tpch:time:rounds} presents typical results for TPC-H queries Q8, Q9, and Q21, and Figure~\ref{fig:ott:time:rounds} presents typical results for the OTT queries.
For each query, the plan in the first round is the original one returned by the optimizer.}


\begin{figure}
\centering
\includegraphics[width=\columnwidth]{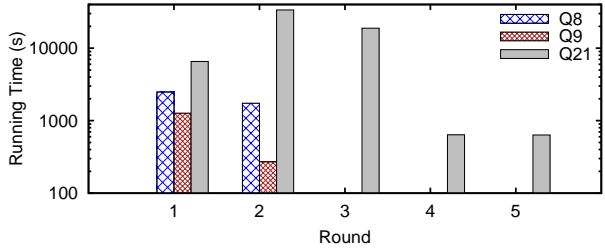}
\caption{Running time of plans generated in re-optimization for TPC-H queries ($z=0$) without calibration of cost units.}
\label{fig:tpch:time:rounds}
\vskip -3ex
\end{figure}

\begin{figure}
\centering
\subfigure[4-join queries]{ \label{fig:ott:time:rounds:4-join}
\includegraphics[width=\columnwidth]{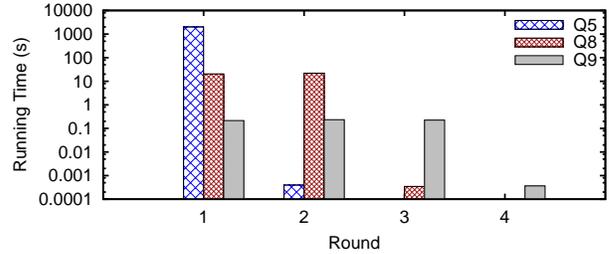}}
\subfigure[5-join queries]{ \label{fig:ott:time:rounds:5-join}
\includegraphics[width=\columnwidth]{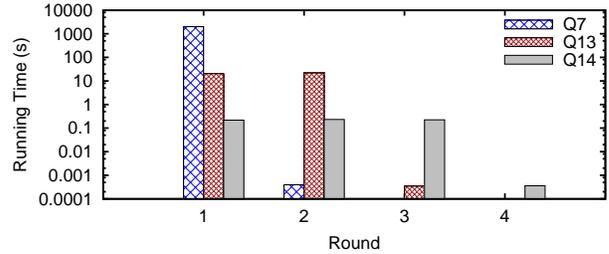}}
\caption{Running time of plans generated in re-optimization for the OTT queries without calibration of cost units.
}
\label{fig:ott:time:rounds}
\vskip -3ex
\end{figure}

{\HL
We have the following observations.
While the plan returned in the second round of iteration (i.e., the first different plan returned by re-optimization) often dramatically reduces the query execution time, it is not always the case.
Sometimes it takes additional rounds before an improved plan is found while the plans generated in between have similar execution times (e.g., 4-join OTT query Q8 and 5-join OTT query Q13).
In the case of TPC-H Q21, the execution times of intermediate plans generated during re-optimization are even not non-increasing.
The plans returned by the optimizer in the second and third round are even much worse than the original plan. (Note that the $y$-axis of Figure~\ref{fig:tpch:time:rounds} is in log scale.)
Fortunately, by continuing with the iteration, we can eventually reach a plan that is much better than the original one.

The above observations give rise to the question why worse plans might be generated during re-optimization.
To understand this, note that in the middle of re-optimization, the plans returned by the optimizer are based on statistics that are \emph{partially} validated by sampling.
The optimizer is likely to generate an ``optimal'' plan due to underestimating the costs of joins that have not been covered by plans generated in previous rounds.
This phenomenon has also been observed in previous work~\cite{ChaudhuriNR08-vldb}.
It is the essential reason that we cannot make Theorem~\ref{theorem:local-optimality} stronger.
As long as there exist uncovered joins, there is no guarantee on the plan returned by the optimizer.
Only upon the convergence of re-optimization we can say that the final plan is locally optimal.

Nevertheless, in practice we can still have various strategies to control the overhead of re-optimization.
For example, we can stop re-optimization if it does not converge after a certain number of rounds, or if the time spent on re-optimization has reached some timeout threshold.
We then simply return the best plan among the plans generated so far, based on their cost estimates by using refined cardinality estimates from sampling~\cite{WuCZTHN13}.
As another option, it might even be worth considering not doing re-optimization at all if the estimated query execution time is shorter than some threshold, or only doing it if we run that plan and get past that threshold without being close to done.
}

\section{Related Work}\label{sec:relatedwork}

Query optimization has been studied for decades, and we refer the readers to~\cite{Chaudhuri98} and~\cite{Ioannidis96} for surveys in this area.

Cardinality estimation is a critical problem in cost-based query optimization, and has triggered extensive research in the database community. Approaches for cardinality estimation in the literature are either static or dynamic.
Static approaches usually rely on various statistics that are collected and maintained periodically in an off-line manner, such as histograms (e.g.,~\cite{Ioannidis-hist03,Poosala-hist97}), samples (e.g.,~\cite{Charikar-sample00,Haas-sample96,Lipton-sample90}), sketches (e.g.,~\cite{Alon-sketch99,Rusu-sketch08}), or even graphical models (e.g.~\cite{GetoorTK01,TzoumasDJ11}). In practice, approaches based on histograms are dominant in the implementations of current query optimizers. However, histogram-based approaches have to rely on the notorious attribute-value-independence (AVI) assumption, and they often fail to capture data correlations, which could result in significant errors in cardinality estimates. While variants of histograms (in particular, multidimensional histograms, e.g.,~\cite{BrunoCG01,MuralikrishnaD88,Poosala-hist97}) have been proposed to overcome the AVI assumption, they suffer from significantly increased overhead on large databases.
{\color{black} Meanwhile, even if we can afford the overhead of using multidimensional histograms, they are still insufficient in many cases, as we discussed in Section~\ref{sec:experiments:multi-histograms}.
Compared with histogram-based approaches, sampling is better at capturing data correlation.
One reason for this is that sampling evaluates queries on real rather than summarized data.}
{\HL
There are many sampling algorithms, and in this paper we just picked a simple one (see~\cite{VengerovMZC15} for a recent survey).
We do not try to explore more advanced sampling techniques (e.g.~\cite{EstanN06,VengerovMZC15}), which we believe could further improve the quality of cardinality estimates.
}

On the other hand, dynamic approaches further utilize information gathered during query runtime. Approaches in this direction include dynamic query plans (e.g.,~\cite{ColeG94,GraefeW89}), parametric query optimization (e.g.~\cite{IoannidisNSS92}), query feedback (e.g.,~\cite{BrunoC02,Stillger-leo}), mid-query re-optimization (e.g.~\cite{KabraD98,MarklRSLP04}), and quite recently, plan bouquets~\cite{DuttH14}.
The ideas behind dynamic query plans and parametric query optimization are similar: rather than picking one single optimal query plan, all possible optimal plans are retained and the decision is deferred until runtime.
%
Both approaches suffer from the problem of combinatorial explosion and are usually used in contexts where expensive pre-compilation stages are affordable. The recent development of plan bouquets~\cite{DuttH14} is built on top of parametric query optimization so it may also incur a heavy query compilation stage.

Meanwhile, approaches based on query feedback record statistics of past queries and use this information to improve cardinality estimates for future queries. Some of these approaches have been adopted in commercial systems such as IBM DB2~\cite{Stillger-leo} and Microsoft SQL Server~\cite{BrunoC02}. Nonetheless, collecting query feedback incurs additional runtime overhead as well as storage overhead of ever-growing volume of statistics.

{\HL
The most relevant work in the literature is the line along mid-query re-optimization~\cite{KabraD98,MarklRSLP04}.
The major difference is that re-optimization was previously carried out at runtime over the actual database rather than at query compilation time over the samples.
The main issue is the trade-off between the overhead spent on re-optimization and the improvement on the query plan.
In the introduction, we have articulated the pros and cons of both techniques.
In some sense, our approach can be thought of as a ``dry run'' of runtime re-optimization.
But it is much cheaper because it is performed over the sampled data.
As we have seen, cardinality estimation errors due to data correlation can be caught by the sample runs.
So it is perhaps an overkill to detect these errors at runtime by running the query over the actual database.
Sometimes optimizers make mistakes that involve a bad ordering that results in a larger than expected result from a join or a selection.
It is true that runtime re-optimization can detect this, but this may require the query evaluator to do a substantial amount of work before it is detected.
For example, if the data is fed into an operator in a non-random order, because of an unlucky ordering the fact that the
operator has a much larger than expected result may not be obvious until a substantial portion of the input has been consumed and substantial
system resources have been expended.
Furthermore, it is non-trivial to stop an operator in mid-flight and switch to a different plan --- for
this reason most runtime query re-optimization deals with changing plans at pipeline boundaries.
Running a bad plan until a pipeline completes might be very expensive.

Nonetheless, we are not claiming that we should just use our approach.
Rather, our approach is complementary to runtime re-optimization techniques: it is certainly possible to apply those techniques on the plan returned by our approach.
The hope is that, after applying our approach, the chance that we need to invoke more costly runtime re-optimizaton techniques can be significantly reduced.
This is similar to the ``alerter'' idea that has been used in physical database tuning~\cite{BrunoC06}, where a lightweight ``alerter'' is invoked to determine the potential performance gain before the more costly database tuning advisor is called.
There is also some recent work on applying runtime re-optimization to Map-Reduce based systems~\cite{KaranasosBKOEXJ14}.
The ``pilot run'' idea explored in this work is close to our motivation, where the goal is also to avoid starting with a bad plan by collecting more accurate statistics via scanning a small amount of samples.
However, currently ``pilot run'' only collects statistics for leaf tables.
It is then interesting to see if ``pilot run'' could be enhanced by considering our technique.
}

In some sense, our approach sits between existing static and dynamic approaches. We combine the advantage of lower overheads from static approaches and the advantage of more optimization opportunities from dynamic approaches. This compromise leads to a lightweight query re-optimization procedure that could bring up better query plans.
{\HL
However, this also unavoidably leads to the coexistence of both histogram-based and sampling-based cardinality estimates during re-optimization, which may cause inconsistency issues~\cite{MarklHKMST07}.
Somehow, this is a general problem when different types of cardinality estimation techniques are used together.
For example, an optimizer that also uses multi-dimensional histograms (in addition to one-dimensional histograms) may also have to handle inconsistencies.
Previous runtime re-optimization techniques are also likely to suffer similar issues.
In this sense, this is an orthogonal problem and we do not try to address it in this paper.
In spite of that, it is interesting future work to see how much more improvement we can get if we further incorporate the approach based on the maximum entropy principle~\cite{MarklHKMST07} into our re-optimization framework to resolve inconsistency in cardinality estimates.
}

Finally, we note that we are not the first that investigates the idea of incorporating sampling into query optimization.
Ilyas et al. proposed using sampling to detect data correlations and then collecting joint statistics for those correlated data columns~\cite{IlyasMHBA04}.
However, this seems to be insufficient if data correlation is caused by specific selection predicates, such as those OTT queries used in our experiments.
Babcock and Chaudhuri also investigated the usage of sampling in developing a robust query optimizer~\cite{BabcockC05}.
While robustness is another interesting goal for query optimizer, it is beyond the scope of this paper.

\section{Conclusion} \label{sec:conclusion}

In this paper, we studied the problem of incorporating sampling-based cardinality estimates into query optimization.
We proposed an iterative query re-optimization procedure that supplements the optimizer with refreshed cardinality estimates via sampling and gives it second chances to generate better query plans.
We show the efficiency and effectiveness of this re-optimization procedure both theoretically and experimentally.

There are several directions worth further exploring.
First, as we have mentioned, the initial plan picked by the optimizer may have great impact on the final plan returned by re-optimization.
While it remains interesting to study this impact theoretically, it might also be an intriguing idea to think about varying the way that the query optimizer works.
For example, rather than just returning one plan, the optimizer could return several candidates and let the re-optimization procedure work on each of them.
This might make up for the potentially bad situation currently faced by the re-optimization procedure that it may start with a bad seed plan.
Second, the re-optimization procedure itself could be further improved.
As an example, note that in this paper we let the optimizer unconditionally accept cardinality estimates by sampling.
However, sampling is by no means perfect.
A more conservative approach is to consider the uncertainty of the cardinality estimates returned by sampling as well.
The previous work~\cite{WuWHN14} has investigated the problem of quantifying uncertainty in sampling-based query running time estimation.
It is very interesting to study a combination of that framework with the re-optimization procedure proposed in this paper.
We leave all these as promising areas for future work.


\newpage

\section{Acknowledgements}

We thank Heng Guo for his help with the proof of Theorem~\ref{theorem:SN-upper-bound}.

{\renewcommand{\baselinestretch}{1.05}
\small

\bibliographystyle{abbrv}
\bibliography{querytime}
}

\appendix
{\HL
\section{Additional Results}\label{appendix:sec:exp}

In this section, we present additional experimental results for the OTT queries.
We also performed the same experiments for TPC-DS queries, and we report the results here.

\subsection{Additional Results on OTT Queries}\label{appendix:sec:exp:ott}

For OTT queries, we also did similar studies regarding the number of plans generated during re-optimization and the time it consumed.
Figure~\ref{fig:hdc:num} presents the number of plans generated during the re-optimization procedure for the 4-join and 5-join OTT queries.
Figure~\ref{fig:hdc-4:time-sampling} and~\ref{fig:hdc-5:time-sampling} further present the comparisons of the running times by excluding or including re-optimization times for these queries (the $y$-axes are in log scale).


\begin{figure}[!h]
\centering
\subfigure[4-join queries]{ \label{fig:hdc:num:4-join}
\includegraphics[width=\columnwidth]{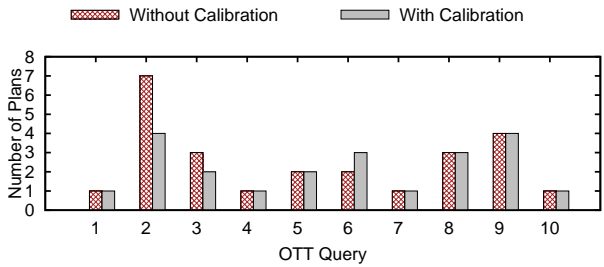}}
\subfigure[5-join queries]{ \label{fig:hdc:num:5-join}
\includegraphics[width=\columnwidth]{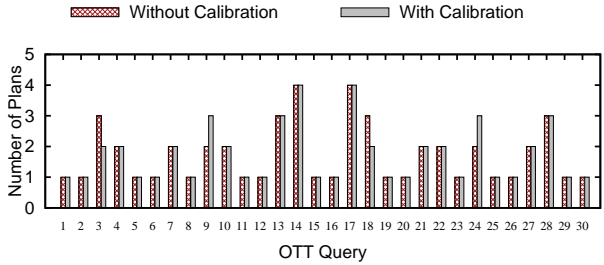}}
\caption{The number of plans generated during re-optimization for the OTT queries.}
\label{fig:hdc:num}
\shrink
\end{figure}

\begin{figure}[htb]
\centering
\subfigure[Without calibration of the cost units]{ \label{fig:hdc-4:time-sampling:no-calib}
\includegraphics[width=\columnwidth]{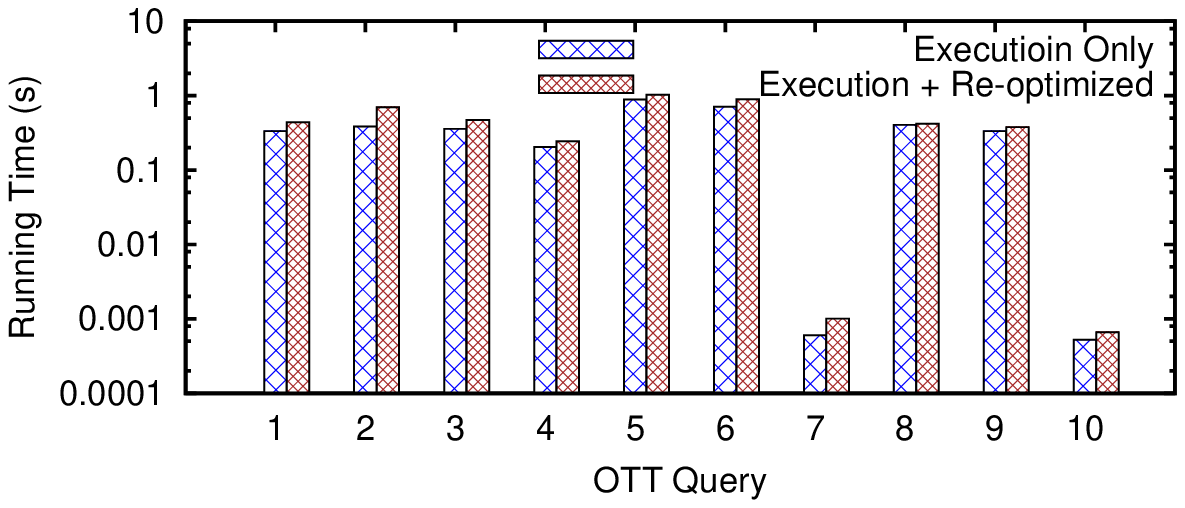}}
\subfigure[With calibration of the cost units]{ \label{fig:hdc-4:time-sampling:calib}
\includegraphics[width=\columnwidth]{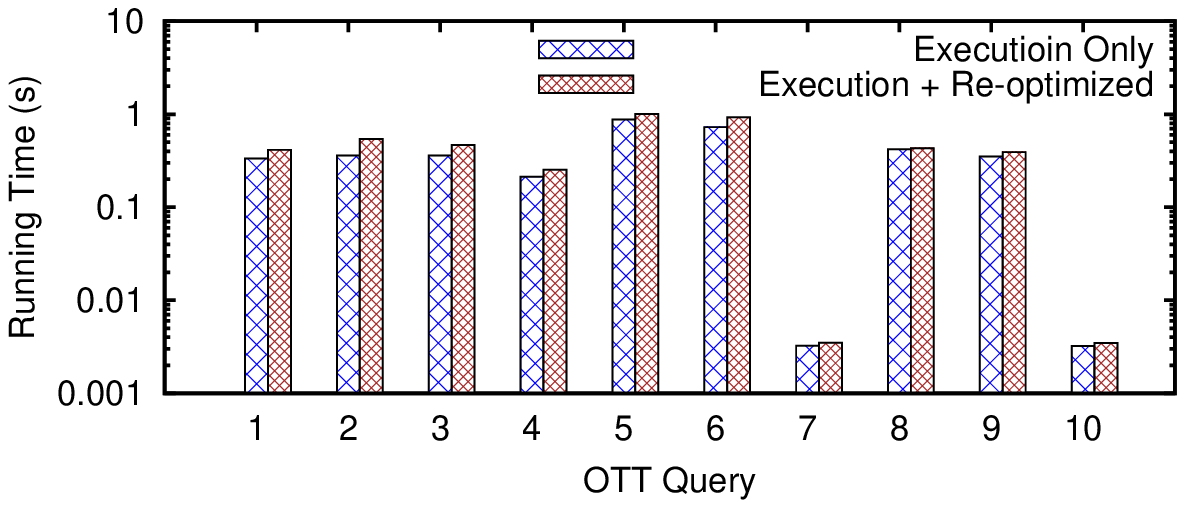}}
\caption{Query running time excluding/including re-optimization time for the 4-join OTT queries.}
\label{fig:hdc-4:time-sampling}
\shrink
\end{figure}

\begin{figure}[!h]
\centering
\subfigure[Without calibration of the cost units]{ \label{fig:hdc-5:time-sampling:no-calib}
\includegraphics[width=\columnwidth]{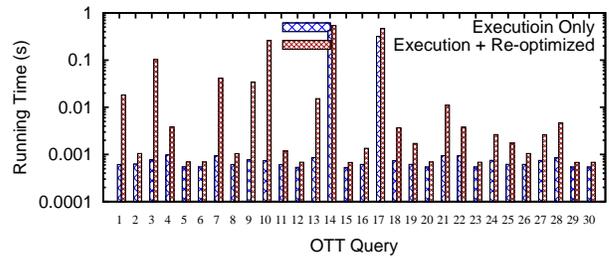}}
\subfigure[With calibration of the cost units]{ \label{fig:hdc-5:time-sampling:calib}
\includegraphics[width=\columnwidth]{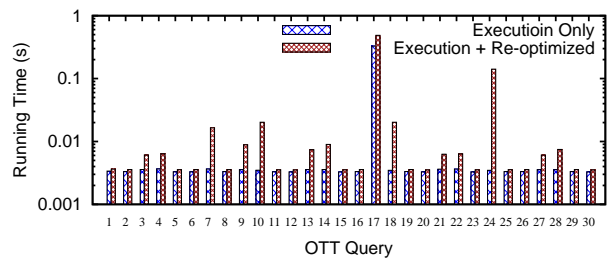}}
\caption{Query running time excluding/including re-optimization time for the 5-join OTT queries.}
\label{fig:hdc-5:time-sampling}
\shrink
\end{figure}

\subsection{Results on the TPC-DS Benchmark}\label{appendix:sec:exp:tpcds}

We conducted the same experiments on the TPC-DS benchmark.
Specifically, we use a subset containing 29 TPC-DS queries that are supported by PostgreSQL and our current implementation,
and we use a 10GB TPC-DS database.
The experiments were performed on a PC with 3.4GHz Intel 4-core CPU
and 8GB memory, and we ran PostgreSQL under Linux 2.6.32.

\begin{figure}[t]
\centering
\subfigure[Without calibration of the cost units]{ \label{fig:tpcds:time:no-calib}
\includegraphics[width=\columnwidth]{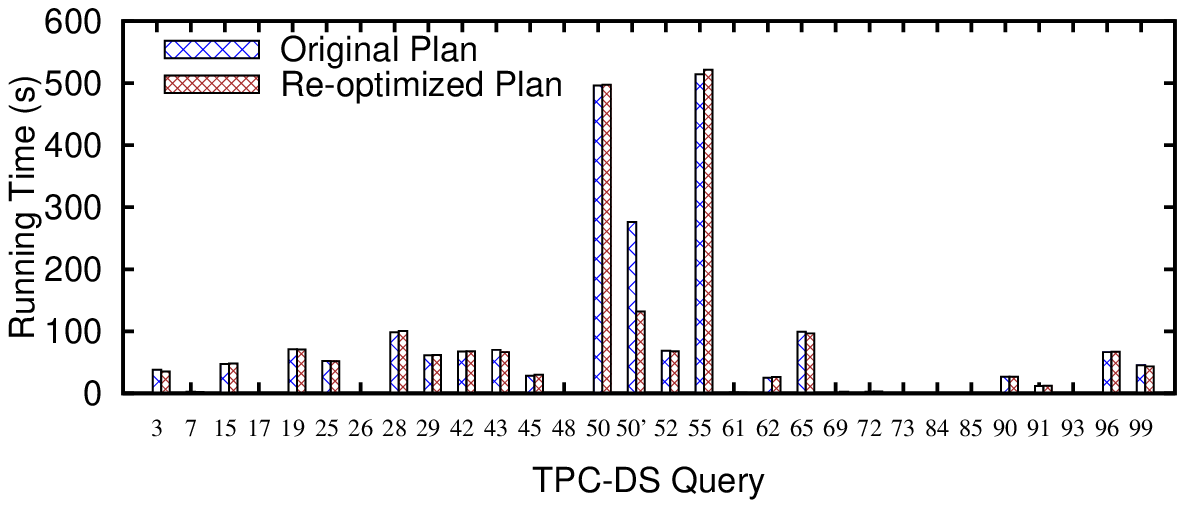}}
\subfigure[With calibration of the cost units]{ \label{fig:tpcds:time:calib}
\includegraphics[width=\columnwidth]{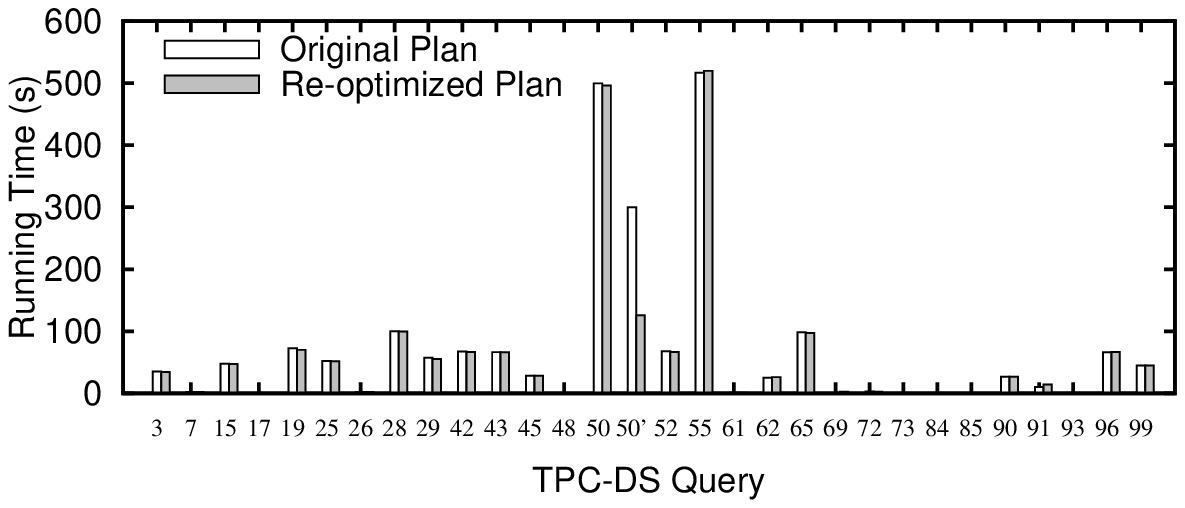}}
\caption{Query running time over 10GB TPC-DS database.}
\label{fig:tpcds:time}
\shrink
\end{figure}

We again use a sampling ratio of 5\%. Figure~\ref{fig:tpcds:time} presents the execution time of each TPC-DS query.
The time spent on running query plans over the samples is again trivial and therefore is not plotted.
Figure~\ref{fig:tpcds:num-plans} further presents the number of plans generated during re-optimization.\footnote{The results we show here are based on the TPC-DS database with only indexes on primary keys of the tables.
The TPC-DS benchmark specification does not recommend any indexes, which potentially could have impact on the optimizer's choices.
Regarding this, we actually further tried to build indexes on columns that are referenced by the queries.
However, the results we obtained are very similar to the ones we present here.}

\begin{figure}
\centering
\includegraphics[width=\columnwidth]{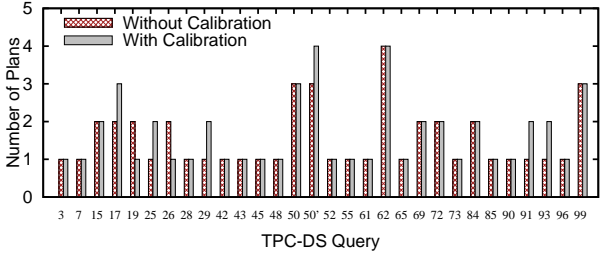}
\caption{The number of plans generated during re-optimization over 10GB TPC-DS database.}
\label{fig:tpcds:num-plans}
\shrink
\end{figure}

Unfortunately, we do not observe remarkable improvement over the TPC-DS queries we tested by using re-optimization.
There are several reasons.
First, the majority of the queries are short-running ones, though some of them are quite complicated in terms of their query plans.
There is little room for improvement on these queries.
Second, for the four relatively long-running queries Q28, Q50, Q55, and Q62 (with running time above 100 seconds), three of them (Q28, Q55, and Q62) are a bit trivial for re-optimization to have impact.
Q28 only accesses one table. Q55 joins a fact table with two small dimension tables (with $<$ 100,000 rows), while Q62 joins a fact table with one small (with 70,000 rows) and three tiny dimension tables (with $<$ 50 rows).
Meanwhile, after checking the plans, we found that all cardinality estimates are on track.
We therefore cannot observe a plan change for any of these three queries (see Figure~\ref{fig:tpcds:num-plans}).
Third, even if there are significant errors in cardinality estimates, they may not lead to significant improvement on the query plan.
This is the case for Q50, which involves joins over two fact tables and three dimension tables.
The execution time is then dominated by the join between the two fact tables.
Unfortunately, the optimizer's cardinality estimate for that join is almost accurate.
As a result, re-optimization cannot change that part of the plan.
It did generate different access paths for the rest of the plan due to detection of cardinality estimation errors, but the remaining execution time after that dominant join was done is trivial.

We actually further tweaked Q50 by modifying the predicates over the dimension tables so that the selectivity estimates were changed.
By doing this, we did identify cases where re-optimization significantly reduced query execution times.
We show one such variant of Q50 (i.e., Q50') in Figure~\ref{fig:tpcds:time}.
The execution time of Q50' dropped from around 300 seconds to around 130 seconds after re-optimization, a 57\% reduction.
Tweaking the other queries may also give more opportunities for re-optimization.
Another way is to further change the distribution that governs data generation.
We do not pursue these options any further, for it remains debatable that these variants might just be contrived.
A more systematic way could be to identify queries that the current optimizers do not handle well, figure out how to handle these queries, then move on to find other kinds of queries that are still not handled well.
We suspect (and hope) that this is one way optimizers will improve over time --- our work on creating the ``optimizer torture test'' in Section~\ref{sec:benchmark} is one step in this direction.
We leave the study of other types of ``difficult'' queries as important direction for future work.

}

\section{Additional Analysis} \label{appendix:sec:analysis}


Continuing with our study of the efficiency of the re-optimization procedure in Section~\ref{sec:analysis:efficiency}, we now consider two special cases where the optimizer significantly overestimates or underestimates the selectivity of a particular join in the returned optimal plan. We focus our discussion on left-deep join trees. Before we proceed, we need the following assumption of cost functions.

\begin{assumption}\label{assumption:monotone}
Cost functions used by the optimizer are monotonically non-decreasing functions of input cardinalities.
\end{assumption}
We are not aware of a cost function that does not conform to this assumption, though.


\subsection{Overestimation Only}

Suppose that the error is an overestimate, that is, the actual cardinality is smaller than the one estimated by the optimizer. Let that join operator be $O$, and let the corresponding subtree rooted at $O$ be $\tree(O)$. Now consider a candidate plan $P$ in the search space, and let $\cost(P)$ be the estimated cost of $P$. Note that the following property holds:
\begin{lemma}\label{lemma:subtree}
$\cost(P)$ is affected (i.e., subject to change) only if $\tree(O)$ is a subtree of $\tree(P)$.
\end{lemma}
Moreover, under Assumption~\ref{assumption:monotone}, the refined cost estimate $\cost'(P)$ satisfies $\cost'(P) < \cost(P)$ because of the overestimate. Therefore, if we let the set of all such plans $P$ be $\mathcal{P}$, then the next optimal plan picked by the optimizer must be from $\mathcal{P}$. We thus can prove the following result:

\begin{theorem}\label{theorem:overestimate}
Suppose that we only consider left-deep join trees. Let $m$ be the number of joins in the query. If all estimation errors are overestimates, then in the worst case the re-optimization procedure will terminate in at most $m + 1$ steps.
\end{theorem}

The proof of Theorem~\ref{theorem:overestimate} is included in Appendix~\ref{appendix:sec:proofs:theorem:overestimate}.
We emphasize that $m+1$ is a worst-case upper bound only. The intuition behind Theorem~\ref{theorem:overestimate} is that, for left-deep join trees, the validated subtree belonging to the final re-optimized plan can grow by at least one more level (i.e., with one more validated join) in each re-optimization step.


\subsection{Underestimation Only}

Suppose that, on the other hand, the error is an underestimate, that is, the actual cardinality is larger than the one estimated by the optimizer. Then things become more complicated: not only those plans that contain the subtree rooted at the erroneous node but the plan with the lowest estimated cost in the rest of the search space (which is not affected by the estimation error) also has the potential of being the optimal plan, after the error is corrected (see Figure~\ref{fig:over-vs-under-estimate}). We next again focus on left-deep trees and present some analysis in such context.

\begin{figure}
\centering
\includegraphics[width=\columnwidth]{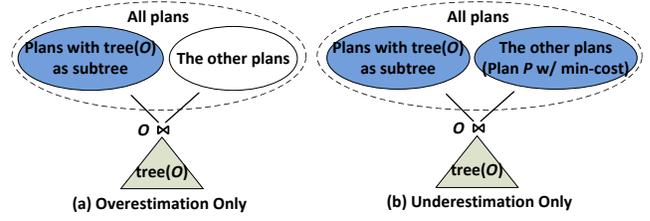}
\caption{Comparison of the space of candidate optimal plans when local estimation errors are overestimation-only or underestimation-only. Candidate optimal plans are shaded.}
\label{fig:over-vs-under-estimate}
\shrink
\end{figure}

Suppose that the join graph $G$ contains $M$ edges. We partition the join trees based on their first joins, which must be the edges of $G$. So we have $M$ partitions in total.

Let $s_i$ be the state when a plan returned by the optimizer uses the $i$-th edge in $G$ ($1\leq i\leq M$) as its first join. The re-optimization procedure can then be thought of as transitions between these $M$ states, i.e., a Markov chain. 
Assume that the equilibrium state distribution $\pi(s_i)$ of this Markov chain exists, subject to 
$$\sum\nolimits_{i=1}^M\pi(s_i)=1.$$
$\pi(s_i)$ therefore represents the probability that a plan generated in the re-optimization procedure would have its first join in the partition $i$ ($1\leq i\leq M$).
We can then estimate the expected number of steps before the procedure terminates as
$$S=\sum\nolimits_{i=1}^M \pi(s_i)\cdot N_i,$$
where $N_i$ is the number of steps/transitions confined in the partition $i$ before termination. Since we only consider left-deep trees, $N_i=N_j$ for $1\leq i,j\leq M$. As a result, $S$ can be simplified as $S=N_i$. So we are left with estimating $N_i$.

$N_i$ is the expected number of steps before termination if the transitions (i.e., re-optimization steps) are confined in the partition $i$. We can divide this process into stages, and each stage contains one more join than the previous stage. If there are $m$ joins, then we will have $m$ stages. In each stage, we go through every plan until we finish, and the average number of plans we considered is $S_{K_j}$ as computed by Equation~\ref{eq:SN}, where $K_j$ is the number of join-trees (actually subtrees) in the stage $j$ ($1\leq j\leq m$). So
$$S=N_i=\sum\nolimits_{j=1}^m S_{K_j}.$$
As an upper bound, we have $S\leq S_{N/M}$, where $N$ is the total number of different join trees considered by the optimizer, and $S_{N/M}$ is computed by Equation~\ref{eq:SN}. This is obtained by applying Theorem~\ref{theorem:efficiency} to each partition. $S_{N/M}$ is usually much smaller than $S_N$, according to Figure~\ref{fig:sN}.
Again, we emphasize that the analysis here only targets worst-case expectations, which might be too pessimistic.

\section{Joint/Marginal Distributions}\label{appendix:sec:distributions}

To see why the straightforward approach that first generates a relation $(A_1, ..., A_K, B_1, ..., B_K)$ based on the joint distribution and then splits it into different relation $R_1(A_1, B_1)$, ..., $R_K(A_K, B_K)$ is incorrect, let us consider the following example.

\begin{example}\label{example:distribution}
Suppose that we need to generate binary values for two attributes $A_1$ and $A_2$, with the following joint distribution:
\begin{enumerate}[(1)]
\item $p_{00}=\Pr(A_1=0, A_2=0)=0.1$;
\item $p_{01}=\Pr(A_1=0, A_2=1)=0$;
\item $p_{10}=\Pr(A_1=1, A_2=0)=0$;
\item $p_{11}=\Pr(A_1=1, A_2=1)=0.9$.
\end{enumerate}
Assume that we generate 10 tuples $(A_1, A_2)$ according to the above distribution. Then we will expect to obtain a multiset of tuples: $\{1\cdot(0,0),9\cdot(1,1)\}$. Here the notation $1\cdot(0,0)$ means there is one $(0,0)$ in the multiset. If we project the tuples onto $A_1$ and $A_2$ (without removing duplicates), we have $A_1=A_2=\{1\cdot 0, 9\cdot 1\}$. Now what is the joint distribution of $(A_1, A_2)$ once we see such a database? Note that we have no idea about the \emph{true} joint distribution that governs the generation of the data, because we are only allowed to see the \emph{marginal} distributions of $A_1$ and $A_2$. A natural inference of the joint distribution may be to consider the cross product $A_1\times A_2$. In this case we have
\begin{eqnarray*}
A_1\times A_2 &=& \{1\cdot 0, 9\cdot 1\}\times\{1\cdot 0, 9\cdot 1\}\\
&=& \{1\cdot(0,0), 9\cdot(0,1), 9\cdot(1,0), 81\cdot(1,1)\}.
\end{eqnarray*}
Hence, the ``observed'' joint distribution of $A_1$ and $A_2$ is:
\begin{enumerate}[(1)]
\item $p'_{00}=\Pr'(A_1=0, A_2=0)=1/100=0.01$;
\item $p'_{01}=\Pr'(A_1=0, A_2=1)=9/100=0.09$;
\item $p'_{10}=\Pr'(A_1=1, A_2=0)=9/100=0.09$;
\item $p'_{11}=\Pr'(A_1=1, A_2=1)=81/100=0.81$.
\end{enumerate}
\end{example}

The ``observed'' joint distribution is indeed ``the'' joint distribution when tables are joined. It might be easier to see this if we rewrite the query in Equation~\ref{eq:query} as
\begin{equation}\label{eq:query2}
\sigma_{A_1=c_1\land\cdots \land A_K=c_K\land B_1=B_2\land\cdots\land B_{K-1}=B_K}(R_1\times\cdots\times R_K).
\end{equation}
The previous example shows that there is information loss when marginalizing out attributes, which is somewhat similar to the notion of lossless/lossy joins in database normalization theory.

\section{Analysis of the OTT Queries}\label{appendix:analysis:ott}

We now compute the query size when Equation~\ref{eq:non-empty-condition} holds.
Consider a specific query $q$ where the constants in the selection predicates are fixed. Let us compute the cardinality of $q$, which is equivalent to computing the selectivity of the predicate in Equation~\ref{eq:query2}. In probabilistic sense, the selectivity $s$ can be interpreted as the chance that a (randomly picked) tuple from $R_1\times\cdots\times R_K$ making the selection predicate true. That is,
$$s=\Pr(A_1=c_1, \cdots, A_K=c_K, B_1=\cdots=B_K).$$

\begin{lemma}\label{lemma:s-actual}
When Equation~\ref{eq:non-empty-condition} holds, we have
$$s=\prod\nolimits_{k=1}^{K}\frac{1}{n(A_k)}.$$
\end{lemma}

The proof of Lemma~\ref{lemma:s-actual} is included in Appendix~\ref{appendix:sec:proofs:lemma:s-actual}. As a result, the size of the query $q$ is
$$|q|=s\cdot \prod\nolimits_{k=1}^K|R_k|=\prod\nolimits_{k=1}^K\frac{|R_k|}{n(A_k)}.$$
To summarize, we have
\begin{displaymath}
|q|=\left\{
\begin{array}{ll}
\prod\nolimits_{k=1}^K\frac{|R_k|}{n(A_k)}, & \textrm{if $c_1=\cdots =c_K$;}\\
0, & \textrm{otherwise.}
\end{array}
\right.
\end{displaymath}

Now what will be the query size estimated by the optimizer? Again, let us compute the estimated selectivity $\hat{s}$, assuming that the optimizer knows the exact histograms of the $A_k$ and $B_k$ ($1\leq k\leq K$). Note that this assumption is stronger than the case in Section~\ref{sec:benchmark:data-gen:postgresql}, where the optimizer possesses exact histograms only for MCV's.
We have the following result:

\begin{lemma}\label{lemma:s-estimated}
Suppose that the AVI assumption is used. Assuming that $\Dom(B_k)$ is the same for $1\leq k\leq K$ and $|\Dom(B_k)|=L$, the estimated selectivity $\hat{s}$ is then
$$\hat{s}=\frac{1}{L^{K-1}}\prod\nolimits_{k=1}^{K}\frac{1}{n(A_k)}.$$
\end{lemma}

The proof of Lemma~\ref{lemma:s-estimated} is included in Appendix~\ref{appendix:sec:proofs:lemma:s-estimated}.
Hence, the estimated size of the query $q$ is
$$\widehat{|q|}=\hat{s}\cdot \prod\nolimits_{k=1}^K|R_k|=\frac{1}{L^{K-1}}\prod\nolimits_{k=1}^K\frac{|R_k|}{n(A_k)}.$$
Note that this is regardless of if Equation~\ref{eq:non-empty-condition} holds or not.
In our experiments (Section~\ref{sec:experiments}) we further used this property to generate instance OTT queries.

Comparing the expressions of $|q|$ and $\widehat{|q|}$, if we define
$$d=\mid|q|-\widehat{|q|}\mid,$$
it then follows that
\begin{displaymath}
d=\left\{
\begin{array}{ll}
\big(1-\frac{1}{L^{K-1}}\big)\cdot\prod\nolimits_{k=1}^K\frac{|R_k|}{n(A_k)}, & \textrm{if $c_1=\cdots =c_K$;}\\
\frac{1}{L^{K-1}}\cdot\prod\nolimits_{k=1}^K\frac{|R_k|}{n(A_k)}, & \textrm{otherwise.}
\end{array}
\right.
\end{displaymath}

Let us further get some intuition about how large $d$ could be by considering the following specific example.
\begin{example}
For simplicity, assume that $M=|R_k| / n(A_k)$ is the same for $1\leq k\leq K$. $M$ is then the number of tuples per distinct value of $A_k$ ($1\leq k\leq K$). In this case,
\begin{displaymath}
d=\left\{
\begin{array}{ll}
\big(1- 1 / L^{K-1}\big)\cdot M^K, & \textrm{if $c_1=\cdots =c_K$;}\\
M^K / L^{K-1}, & \textrm{otherwise.}
\end{array}
\right.
\end{displaymath}
If $L=100$ and $M=100$, it then follows that
\begin{displaymath}
d=\left\{
\begin{array}{ll}
100^K - 100, & \textrm{if $c_1=\cdots =c_K$;}\\
100, & \textrm{otherwise.}
\end{array}
\right.
\end{displaymath}
For $K=4$ (i.e., just 3 joins), $d\approx 10^8$ if it happens to be that $c_1=\cdots =c_K$.
Moreover, if $c_1=\cdots =c_K$, then $|q|>\widehat{|q|}$.
Therefore, the optimizer will significantly underestimate the size of the join (by $10^8$).
So it is likely that the optimizer would pick an inefficient plan that it thought were efficient.
\end{example}

\section{Formal Definitions}

In this section, we further formalize the notion of local/global transformations.
This formalization is useful in some proofs presented in Appendix~\ref{appendix:sec:proofs}.

\subsection{Structural Equivalence}

We start by introducing structural equivalence, which is a special case of local transformation.
While we still use $\tree(P)$ to denote the join tree of a query plan $P$, here $\tree(P)$ is interpreted as a tree rather than a set.

\begin{definition}[Structural Equivalence]
Two plans $P$ and $P'$ (of the same query) are \emph{structurally equivalent} if $\tree(P)$ and $\tree(P')$ are the same.
\end{definition}

In other words, structurally equivalent plans have the same join order of the relations (i.e., $A\bowtie B$ is different from $B\bowtie A$); they only differ in the specific choices of physical operators (e.g., hash join vs. sort-merge join).

\subsection{Local and Global Transformations} \label{sec:analysis:transformations}

We further represent a join tree with a bottom-up, left-to-right encoding. Consider the following example.
\begin{example} [Encoding of Join Trees] \label{ex:join-tree-encoding}
Suppose that we have a join query over four relations $A$, $B$, $C$, and $D$, with the following join graph
$$G=\{(A, B), (C, D), (A, C), (A, D)\}.$$
In Figure~\ref{fig:transformations}, we have presented one left-deep tree plan $T_1$ and one bushy tree plan $T_2$ of the query.
The encodings of $T_1$ and $T_2$ are:
\begin{enumerate}[($T_1$)]
\item $((A\bowtie B)\bowtie C)\bowtie D \encoding (AB, ABC, ABCD)$;
\item $(A\bowtie B)\bowtie (C\bowtie D)\encoding (AB, CD, ABCD)$.
\end{enumerate}
\end{example}

We use $\code(T)$ to represent the encoding of a join tree $T$.
We next formalize local/global transformations based on the above encoding scheme.
\begin{definition}[Formalization]
Let $T$ and $T'$ be two join trees (of the same query). $T'$ is a \emph{local} transformation of $T$ if the following conditions hold:
\begin{enumerate}[(1)]
\item $\forall J\in\code(T)$, there is $J'\in\code(T')$ such that $J'$ is a permutation of $J$;
\item $\code(T')$ is itself a permutation of $\code(T)$.
\end{enumerate}
Otherwise, $T'$ is a \emph{global} transformation of $T$.
\end{definition}

For example, in Figure~\ref{fig:transformations} the encodings of the two local transformations $T'_1$ and $T'_2$ (of $T_1$ and $T_2$, respectively) are:
\begin{enumerate}[($T'_1$)]
\item $(C\bowtie(A\bowtie B))\bowtie D \encoding (AB, CAB, CABD)$;
\item $(C\bowtie D)\bowtie (A\bowtie B) \encoding (CD, AB, CDAB)$.
\end{enumerate}

Recall that, given two plans $P$ and $P'$ (of the same query), $P'$ is a local/global transformation of $P$ if $\tree(P')$ is a local/global transformation of $\tree(P)$. Since by definition a join tree is a local transformation of itself, it then holds that
\begin{lemma}
If two plans $P$ and $P'$ are structurally equivalent, then $P'$ is a local transformation of $P$.
\end{lemma}

\section{Proofs} \label{appendix:sec:proofs}

This section includes proofs of the theoretical results mentioned in the paper that were omitted due to space constraint.

%
%
%

\subsection{Proof of Theorem~\ref{theorem:sequence-transformation}} \label{appendix:sec:proofs:theorem:sequence-transformation}

\begin{proof}
First note that the three cases are mutually exclusive. Next, they are also complete. To see this, note that the only situation not covered is that the re-optimization procedure terminates after $n + 1$ steps ($n > 1$) and during re-optimization there exists some $j < i$ ($1\leq i \leq n - 1$) such that $P_i$ is a local transformation of $P_j$. But this is impossible, because by Corollary~\ref{corollary:local-transformation} the re-optimization procedure would then terminate after $i + 1$ steps. Since $i\leq n - 1$, we have $i + 1\leq n$, which is contradictory with the assumption that the re-optimization procedure terminates after $n + 1$ steps. Therefore, we have shown the completeness of the three cases stated in Theorem~\ref{theorem:sequence-transformation}.

We are left with proving that the three cases could happen when the re-optimization procedure terminates. Case (1) and (2) are clearly possible, while Case (3) is implied by Corollary~\ref{corollary:local-transformation}. This completes the proof of the theorem.
\end{proof}

\subsection{Proof of Lemma~\ref{lemma:expected-steps}} \label{appendix:sec:proofs:lemma:expected-steps}

\begin{proof}
The procedure terminates when the head ball is marked. Since a marked ball is uniformly placed at any position in the queue, the probability that a ball is marked after $k$ steps is $k/N$ ($1\leq k \leq N$). So is the probability that the head ball is marked. Formally, let $A_k$ be the event that the head ball is marked after exactly $k$ steps. Then $\Pr(A_1)=1/N$ and $\Pr(A_k|\bar{A}_1\cap\cdots\cap\bar{A}_{k-1})=k/N$. As a result, letting $B_k=\bar{A}_1\cap\cdots\cap\bar{A}_{k-1}$ we have
\begin{eqnarray*}
\Pr(A_k)&=&\Pr(A_k|B_k)\Pr(B_k)+\Pr(A_k|\bar{B}_k)\Pr(\bar{B}_k)\\
&=&\frac{k}{N}\cdot\Pr(B_k) + 0\cdot\Pr(\bar{B}_k)\\
&=&\frac{k}{N}\cdot\Pr(B_k).
\end{eqnarray*}
Now let us calculate $\Pr(B_k)$. We have
\begin{eqnarray*}
\Pr(B_k)&=&\Pr(\bar{A}_1\cap\cdots\cap\bar{A}_{k-1})\\
&=&\Pr(\bar{A}_{k-1}|\bar{A}_1\cap\cdots\cap\bar{A}_{k-2})\Pr(\bar{A}_1\cap\cdots\cap\bar{A}_{k-2})\\
&=&(1-\frac{k-1}{N})\Pr(B_{k-1}).
\end{eqnarray*}
Therefore, we now have a recurrence equation for $\Pr(B_k)$ and thus,
\begin{eqnarray*}
\Pr(B_k)&=&(1-\frac{k-1}{N})\cdots(1-\frac{2}{N})\Pr(B_2)\\
&=&(1-\frac{k-1}{N})\cdots(1-\frac{2}{N})\Pr(\bar{A}_1)\\
&=&(1-\frac{k-1}{N})\cdots(1-\frac{2}{N})(1-\frac{1}{N}).
\end{eqnarray*}
The expected number of steps that the procedure would take before its termination is then
\begin{eqnarray*}
S_N&=&\sum_{k=1}^N k\cdot\Pr(A_k)\\
&=&\sum_{k=1}^N k\cdot(1-\frac{1}{N})\cdots(1-\frac{k-1}{N})\cdot\frac{k}{N}.
\end{eqnarray*}
This completes the proof of the lemma.
\end{proof}

\subsection{Proof of Theorem~\ref{theorem:SN-upper-bound}} \label{appendix:sec:proofs:theorem:SN-upper-bound}

We need the following lemma before we prove Theorem~\ref{theorem:SN-upper-bound}.

\begin{lemma}\label{lemma:Xk-upper-bound}
Let the $k$-th summand in $S_N$ be $X_k$, i.e.,
\begin{eqnarray*}
X_k &=& k\cdot(1-\frac{1}{N})\cdots(1-\frac{k-1}{N})\cdot\frac{k}{N}\\
&=& \frac{N!}{(N-k)!}\cdot\frac{1}{N^k}\cdot\frac{k^2}{N}.
\end{eqnarray*}
For any $k\geq N^{1/2 + \epsilon}$ ($\epsilon > 0$), we have $X_k=O(e^{-N^{2\epsilon}})$.
\end{lemma}

\begin{proof}
We prove this in two steps:
\begin{enumerate}[(i)]
\item if $N^{1/2 + \epsilon} \leq k \leq N/2$, then $X_k=O(e^{-N^{2\epsilon}})$;
\item if $k \geq N/2$, then $X_{k+1}<X_k$ for sufficiently large $N$ ($N\geq 5$ indeed).
\end{enumerate}

We prove (i) first. Consider $\ln X_k$. We have
$$\ln X_k = \ln \big(N!\big) -\ln \big((N-k)!\big) + 2\ln k - (k+1)\ln N.$$
By using Stirling's formula,
$$\ln\big(N!\big)=N\ln N - N + O(\ln N).$$
It then follows that
\begin{eqnarray*}
\ln X_k &=& (N-k)\ln N  - (N - k)\ln (N - k) - k + O(\ln N)\\
&=&(N - k)\ln\big(\frac{N}{N - k}\big) - k + O(\ln N)\\
&=&(N - k)\ln\big(1 + \frac{k}{N - k}\big) - k + O(\ln N).
\end{eqnarray*}
Using Taylor's formula for $f(x)=\ln(1+x)$, we have
$$\ln\big(1 + \frac{k}{N - k}\big) = \frac{k}{N - k}-\frac{1}{2}\cdot\frac{k^2}{(N-k)^2}+O\big(\frac{1}{3}\cdot\frac{k^3}{(N-k)^3}\big).$$
Therefore, it follows that
$$\ln X_k = -\frac{1}{2}\cdot\frac{k^2}{N-k}+O\big(\frac{1}{3}\cdot\frac{k^3}{(N-k)^2}\big)+O(\ln N).$$
Since $k\leq N/2$, $\frac{k}{N-k}\leq 1$. As a result,
$$\frac{k^3}{(N-k)^2}=\frac{k}{N-k}\cdot\frac{k^2}{N-k}\leq\frac{k^2}{N-k}.$$
On the other hand, since $k\geq N^{1/2 + \epsilon}$, $k^2\geq N^{1+2\epsilon}$ and thus,
$$\frac{k^2}{N-k}>\frac{k^2}{N}\geq N^{2\epsilon}>O(\ln N).$$
Therefore, $\ln X_k$ is dominated by the term $\frac{k^2}{N-k}$. Hence
$$\ln X_k=O(-\frac{k^2}{N-k})=O(-N^{2\epsilon}),$$
which implies $X_k=O(e^{-N^{2\epsilon}})$.

We next prove (ii). We consider the ratio of $X_{k+1}$ and $X_k$, which gives us the following:
$$r_k=\frac{X_{k+1}}{X_k}=\frac{N-k}{N}\cdot\frac{(k+1)^2}{k^2}=\big(1-\frac{k}{N}\big)\cdot\big(1+\frac{1}{k}\big)^2.$$
Since $k\geq N/2$, $1-\frac{k}{N}\leq 1/2$ and $\frac{1}{k}\leq 2/N$. It follows that
$$r_k\leq\frac{1}{2}\cdot\big(1+\frac{2}{N}\big)^2.$$
Letting $r_k < 1$ gives $N > 2(\sqrt{2} + 1)\approx 4.83$. So we have $r_k < 1$ when $N\geq 5$, which concludes the proof of the lemma.
\end{proof}

We are ready to give a proof to Theorem~\ref{theorem:SN-upper-bound}.
\begin{proof} (of Theorem~\ref{theorem:SN-upper-bound})
Let $X_k$ be the same as that defined in Lemma~\ref{lemma:Xk-upper-bound}. Then $S_N=\sum_{k=1}^{N} X_k$.
According to Lemma~\ref{lemma:Xk-upper-bound}, $X_k=O(e^{-N^{2\epsilon}})$ if $k\geq N^{1/2 + \epsilon}$.
It then follows that
$$\sum\nolimits_{k\geq N^{1/2 + \epsilon}} X_k = O(N\cdot e^{-N^{2\epsilon}})=o(\sqrt{N}).$$
On the other hand, when $k < N^{1/2 + \epsilon}$, we can pick a sufficiently small $\epsilon$ so that $k\leq N^{1/2}$. If this is the case, then $k^2\leq N$ and thus $\frac{k^2}{N}\leq 1$. As a result, $X_k\leq 1$ and hence,
$$\sum\nolimits_{k < N^{1/2 + \epsilon}} X_k \leq \sqrt{N}.$$
It then follows that
$$S_N=\sum\nolimits_{k\geq N^{1/2 + \epsilon}} X_k + \sum\nolimits_{k < N^{1/2 + \epsilon}} X_k = O(\sqrt{N}),$$
which completes the proof of the theorem.

To pick such an $\epsilon$, note that it is sufficient if $\epsilon$ satisfies $N^{1/2+\epsilon}\leq \lfloor N^{1/2} \rfloor + 1$. This gives
$$\epsilon \leq \ln\big(1 + \lfloor\sqrt{N}\rfloor\big) / \ln N - 1/2.$$
Note that the right hand side must be greater than 0, because 
$$\ln\big(1 + \lfloor\sqrt{N}\rfloor\big) > \ln \sqrt{N}.$$
As an example, if $N=100$, then $\epsilon\leq 0.0207$.
\end{proof}

\subsection{Proof of Theorem~\ref{theorem:local-optimality}} \label{appendix:sec:proofs:theorem:local-optimality}

\begin{proof}
First note that $\cost^s(P_i)$ is well defined for $1\leq i\leq n$, because all these plans have been validated via sampling.
As a result, all the statistics regarding $P_1$, ..., $P_{n-1}$ were already included in $\Gamma$ when the re-optimization procedure returned $P_n$ as the optimal plan. Since $P_n$ is the final plan, it implies that it did not change in the final round of iteration, where $\Gamma$ also included statistics of $P_n$. Then the only reason for the optimizer to pick $P_n$ as the optimal plan is $\cost^s(P_n)\leq \cost^s(P_i)$ ($1\leq i \leq n - 1$).
\end{proof}

\subsection{Proof of Corollary~\ref{corollary:local-optimality:overestimates}} \label{appendix:sec:proofs:corollary:local-optimality:overestimates}

\begin{proof}
The argument is similar to that used in the proof of Theorem~\ref{theorem:local-optimality}. When the optimizer returned $P_{i+1}$ as the optimal plan, it had already seen the statistics of $P_i$ after sampling. The only reason it chose $P_{i+1}$ is then $\cost^o(P_{i+1})\leq\cost^s(P_i)$. Here $\cost^o(P_{i+1})$ is the original cost estimate of $P_{i+1}$ (perhaps based on histograms), because at this point sampling has not been used for $P_{i+1}$ yet. But given that only overestimation is possible during re-optimization, it follows that $\cost^s(P_{i+1})\leq \cost^o(P_{i+1})$. As a result, $\cost^s(P_{i+1})\leq\cost^s(P_i)$. Note that $i$ is arbitrary in the argument so far. Thus we have proved the theorem.
\end{proof}

\subsection{Proof of Theorem~\ref{theorem:optimality-local-transformation}}
\label{appendix:sec:proofs:theorem:optimality-local-transformation}

\begin{proof}
The proof is straightforward. By Theorem~\ref{theorem:sequence-transformation}, local transformation is only possible in the last second step of re-optimization. At this stage, $\Gamma$ has already included all statistics regarding $P$ and $P'$: actually they share the same joins given that they are local transformations. Hence, $\cost^s(P)$ and $\cost^s(P')$ are both known to the optimizer. Then $\cost^s(P)\leq\cost^s(P')$ must hold because the optimizer decides that $P$, rather than $P'$, is the optimal plan.
\end{proof}

\subsection{Proof of Lemma~\ref{lemma:s-actual}}
\label{appendix:sec:proofs:lemma:s-actual}

\begin{proof}
We have $s = \Pr(A_1=c_1, \cdots, A_K=c_K, B_1=\cdots=B_K)$, or equivalently,
\begin{eqnarray*}
s&=&\Pr(B_1=\cdots =B_K|A_1=c_1, \cdots, A_K=c_K)\\
&\times &\Pr(A_1=c_1, \cdots , A_K=c_K).
\end{eqnarray*}

For notational convenience, define
$$\Pr(\mathbf{B}|\mathbf{A})=\Pr(B_1=\cdots =B_K|A_1=c_1,\cdots,A_K=c_K),$$
and similarly define
$$\Pr(\mathbf{A})=\Pr(A_1=c_1, \cdots , A_K=c_K).$$
Since $B_k=A_k$ ($1\leq k\leq K$), $\Pr(\mathbf{B}|\mathbf{A})=1$ when Equation~\ref{eq:non-empty-condition} holds.
Therefore, $s=\Pr(\mathbf{A})$.
Moreover, since we generate $R_k$ independently and $\Pr(A_k)$ is uniform ($1\leq k\leq K$), it follows that
$$s=\Pr(\mathbf{A})=\prod\nolimits_{k=1}^{K}\Pr(A_k=c_k)=\prod\nolimits_{k=1}^{K}\frac{1}{n(A_k)},$$
where $n(A_k)$ is the number of distinct values of $A_k$ as before.
This completes the proof of the lemma.
\end{proof}

\subsection{Proof of Lemma~\ref{lemma:s-estimated}}
\label{appendix:sec:proofs:lemma:s-estimated}

\begin{proof}
Assume that $\Dom(B_k)=\{C_1,...,C_L\}$ for $1\leq k\leq K$.
By definition of selectivity, we have
$$\hat{s}=\Pr(A_1=c_1, \cdots, A_K=c_K, B_1=B_2, \cdots, B_{K-1}=B_K).$$
According to the AVI assumption, it follows that
$$\hat{s}=\prod\nolimits_{k=1}^{K}\Pr(A_k=c_k)\times\prod\nolimits_{k=1}^{K-1}\Pr(B_k=B_{k+1}).$$
Let us consider $\Pr(B_k=B_{k+1})$. We have
\begin{eqnarray*}
\Pr(B_k=B_{k+1})&=&\sum_{l=1}^{L}\sum_{l'=1}^{L}\big(\Pr(B_k=C_l, B_{k+1}=C_{l'})\\
&\times&\Pr(B_k=B_{k+1}|B_k=C_l, B_{k+1}=C_{l'})\big).
\end{eqnarray*}
Since $\Pr(B_k=B_{k+1}|B_k=C_l, B_{k+1}=C_{l'})=1$ if and only if $l=l'$ (otherwise it equals 0), it follows that
$$\Pr(B_k=B_{k+1})=\sum\nolimits_{l=1}^{L}\Pr(B_k=C_l, B_{k+1}=C_l).$$
Moreover, since we have assumed that the optimizer knows exact histograms and $\Pr(B_k)$ is uniform for $1\leq k\leq K$ (recall that $B_k=A_k$ and $\Pr(A_k)$ is uniform),
\begin{eqnarray*}
\Pr(B_k=C_l, B_{k+1}=C_l)&=&\frac{n(B_k=C_l, B_{k+1}=C_l)}{|B_k|\cdot|B_{k+1}|}\\
&=&\frac{n(B_k=C_l)}{|B_k|}\cdot\frac{n(B_{k+1}=C_l)}{|B_{k+1}|}\\
&=&\frac{1}{n(B_k)\cdot n(B_{k+1})}.
\end{eqnarray*}
Because $n(B_k)=L$ for $1\leq k\leq K$, it follows that
$$\Pr(B_k=B_{k+1})=L\times\frac{1}{L^2}=\frac{1}{L},$$
and as a result,
$$\hat{s}=\frac{1}{L^{K-1}}\prod\nolimits_{k=1}^{K}\Pr(A_k=c_k)=\frac{1}{L^{K-1}}\prod\nolimits_{k=1}^{K}\frac{1}{n(A_k)}.$$
This completes the proof of the lemma.
\end{proof}

\subsection{Proof of Theorem~\ref{theorem:overestimate}} \label{appendix:sec:proofs:theorem:overestimate}

\begin{proof}
We use $P_i$ to denote the plan returned in the $i$-th step, and use $O_i$ to denote the \emph{lowest} join in $P_i$ where an overestimate occurs. We further use $I(O_i)$ to denote the index of $O_i$ in the encoding of $\tree(P_i)$ (ref. Section~\ref{sec:analysis:transformations}).

Let $\mathcal{P}_i$ be the plans whose join trees contain $\tree(O_i)$ as a subtree. By Lemma~\ref{lemma:subtree}, $P_{i+1}\in\mathcal{P}_i$. Moreover, we must have $I(O_{i+1}) > I(O_i)$, that is, the lowest join in $P_{i+1}$ with an overestimate must be at a higher level than that in $P_i$, because $\tree(O_i)$ is a subtree of $\tree(P_{i+1})$ and $\tree(O_i)$  has been validated when validating $P_i$. Note that we can only have at most $m$ different $I(O_i)$'s given that the query contains $m$ joins.
Therefore, after (at most) $m$ steps, we must have $\mathcal{P}_m=\emptyset$. As a result, by Theorem~\ref{theorem:sequence-transformation} we can only have at most one additional \emph{local} transformation on top of $P_m$.
But this cannot occur for left-deep trees, because $\tree(P_{m-1})=\tree(P_{m})$ if the worst case really occurs (i.e., in every step an overestimate occurs such that $I(O_{i+1}) = I(O_i) + 1$). So any local transformation must have been considered when generating $P_m$ by using validated cardinalities from $P_{m-1}$. Hence the re-optimization procedure terminates in at most $m+1$ steps: the $(m+1)$-th step just checks $P_{m+1}=P_m$, which must be true.
\end{proof}

\section{Further Notes}

Although we find that benchmarks may not be the right target to show the superiority of re-optimization, as many query optimizers are designed to perform well on them, it might be possible to look at queries executed by open-source software, or report builders to come up with more real examples of queries where our proposed approach can make a difference.
Unfortunately, so far we are not aware of the availability of such databases and queries.

Another interesting question we have briefly discussed in the paper is the optimality of the final plan returned by re-optimization.
However, a thorough analysis is not possible unless we list all plans enumerated by the optimizer and test their running times. This
then requires significant change to the optimizer. Even though the optimizer could support this functionality, as was pointed out by Gu et al.~\cite{GuSW12}, due to the large search space, timing the execution of each single plan in the search space is in general impossible.
Nevertheless, our analysis suggests that the re-optimized plan is better than the original plan, even if it may still be a bad plan.
Somehow, ``optimization'' is actually a misnomer in the context of query ``optimization'' --- query optimizers do not find optimal plans, they try to avoid bad ones. Our ``optimization'' is in this spirit.

Furthermore, while it is beyond the scope of this paper, it remains interesting to compare the plans generated by our proposed technique and previous work on runtime query re-optimization. Given that runtime query re-optimization techniques can observe the true cardinalities, we would expect that they can at least generate plans that are as good as ones returned by our approach. The intriguing question is how much more improvement they can make. Since PostgreSQL does not support runtime query re-optimization, it requires significant engineering effort to come up with an implementation. We thus leave this investigation for future work.



\end{document}